\pdfoutput=1

\documentclass[a4paper,UKenglish,cleveref, autoref, thm-restate,authorcolumns]{lipics-v2019}

\newcommand{\conferenceorfull}[2]{#2}

\title{Decision Problems in Information Theory} 

\titlerunning{Decision Problems in Information Theory} 

\author{Mahmoud Abo Khamis}{relational\underline{AI}, Berkeley, CA, USA}{}{}{}
\author{Phokion G.\ Kolaitis}{UC Santa Cruz and IBM Research - Almaden,  CA, USA}{}{}{Research partially supported by NSF Grant IIS-1814152.}
\author{Hung Q.\ Ngo}{relational\underline{AI}, Berkeley, CA, USA}{}{}{}
\author{Dan Suciu}{University of Washington, Seattle, WA,
  USA}{}{}{Research partially supported by NSF IIS-1907997, NSF III-1703281, NSF III-1614738, and NSF AiTF-1535565}

\authorrunning{M.\,Abo Khamis, P.\,G. Kolaitis, H.\,Q. Ngo, and D. Suciu} 

\Copyright{Mahmoud Abo Khamis, Phokion G.\ Kolaitis, Hung Q.\ Ngo, and Dan Suciu} 



\conferenceorfull{\category{Track B: Automata, Logic, Semantics, and Theory of Programming}}{} 

\conferenceorfull{\relatedversion{A full version of the paper is available at \url{https://arxiv.org/abs/2004.08783}}}{}

\acknowledgements{We thank Miika Hannula for several useful pointers to earlier work on the implication problem for conditional independence.}

\nolinenumbers 

\conferenceorfull{}{\hideLIPIcs}

\EventEditors{Artur Czumaj, Anuj Dawar, and Emanuela Merelli}
\EventNoEds{3}
\EventLongTitle{47th International Colloquium on Automata, Languages, and Programming (ICALP 2020)}
\EventShortTitle{ICALP 2020}
\EventAcronym{ICALP}
\EventYear{2020}
\EventDate{July 8--11, 2020}
\EventLocation{Saarbrücken, Germany (virtual conference)}
\EventLogo{}
\SeriesVolume{168}
\ArticleNo{106}

\usepackage{MnSymbol}
\usepackage{tikz,tikz-qtree,tikz-qtree-compat,tikz-3dplot}

\usetikzlibrary{shapes,decorations.markings,arrows.meta,fit}

\usepackage{todonotes}
\usepackage{url}
\usepackage{bm}

\colorlet{LightViolet}{violet!40}
\colorlet{LightRed}{red!40}
\colorlet{LightOrange}{orange!40}
\colorlet{LightGreen}{green!40}
\colorlet{LightBlue}{blue!40}
\colorlet{DarkGreen}{green!50!black}
\colorlet{DarkRed}{red!70!black}
\colorlet{DarkCyan}{red!70!black}
\colorlet{DarkBlue}{blue!80!black}
{\definecolor{DarkOrange}{rgb}{1.0, 0.49, 0.0}
\definecolor{Airforceblue}{rgb}{0.36, 0.54, 0.66}

\newcommand{\defeq}{\stackrel{\mathrm{def}}{=}}

\renewcommand{\hom}{\text{\sf HOM}}

\newcommand{\calF}{\mathcal F}

\newcommand{\Z}{\mathbb Z} 
\newcommand{\N}{\mathbb N} 
\newcommand{\R}{\mathbb R} 
\newcommand{\Q}{\mathbb Q} 

\newcommand{\be}{\begin{enumerate}}
\newcommand{\ee}{\end{enumerate}}
\newcommand{\bi}{\begin{itemize}}
\newcommand{\ei}{\end{itemize}}
\newcommand{\beq}{\begin{equation}}
\newcommand{\eeq}{\end{equation}}

\newcommand{\bp}{\begin{proof}}
\newcommand{\ep}{\end{proof}}
\newcommand{\bcor}{\begin{corollary}}
\newcommand{\ecor}{\end{corollary}}
\newcommand{\bthm}{\begin{theorem}}
\newcommand{\ethm}{\end{theorem}}
\newcommand{\blmm}{\begin{lemma}}
\newcommand{\elmm}{\end{lemma}}
\newcommand{\bdefn}{\begin{definition}}
\newcommand{\edefn}{\end{definition}}
\newcommand{\bprop}{\begin{proposition}}
\newcommand{\eprop}{\end{proposition}}
\newcommand{\brmk}{\begin{remark}}
\newcommand{\ermk}{\end{remark}}
\newcommand{\bclm}{\begin{claim}}
\newcommand{\eclm}{\end{claim}}
\newcommand{\bex}{\begin{example}}
\newcommand{\eex}{\end{example}}

\newcommand{\norm}[1]{\|#1\|}

\newcommand{\suchthat}{\ | \ }

\newcommand{\vspan}{\mathop{\textnormal{span}}}

\newcommand{\rank}[1]{\mathop{\textnormal{rank}} #1}
\newcommand{\mv}[1]{\bm{#1}}

\newtheorem{fact}[theorem]{Fact}

\definecolor{Red}{RGB}{255,204,204}
\definecolor{Green}{RGB}{204,255,204}
\definecolor{Blue}{RGB}{204,204,255}

\newcommand{\set}[1]{\{#1\}}                    
\newcommand{\setof}[2]{\{{#1}\mid{#2}\}}        

\newcommand{\convex}{\text{\sf conv}}
\newcommand{\cone}{\text{\sf cone}}
\newcommand{\con}{\text{\sf cone}}
\newcommand{\cl}[1]{\overline{{#1}}}

\newcommand{\ebic}{\text{\sf EBIC}}
\newcommand{\aebic}{\text{\sf AEBIC}}

\newcommand{\iip}{\text{\sf IIP}}
\newcommand{\econdiip}{\text{\sf ECIIP}}
\newcommand{\aecondiip}{\text{\sf AECIIP}}
\newcommand{\maxiip}{\text{\sf MaxIIP}}
\newcommand{\entropic}{\Gamma_n^*}
\newcommand{\aentropic}{\overline \Gamma_n^*}

\newcommand{\eat}[1]{}

\begin{document}

\maketitle

\begin{abstract}
  Constraints on entropies are considered to be the laws of information theory.  Even though the pursuit of their discovery has been a central theme of research in information theory, the algorithmic aspects of constraints on entropies remain largely unexplored.  Here, we initiate an investigation of decision problems about constraints on entropies by placing several different such problems into levels of the arithmetical hierarchy.  We establish the following results on checking the validity over all almost-entropic functions:  first, validity of a Boolean information constraint arising from a monotone Boolean formula is co-recursively enumerable; second,  validity of ``tight'' conditional information constraints is in $\Pi^0_3$. Furthermore,  under some restrictions, validity of conditional information constraints ``with slack'' is in $\Sigma^0_2$, and  validity of  information inequality constraints involving max is Turing equivalent to validity of  information inequality constraints (with no max involved).  We also prove that the classical implication problem for conditional independence statements is co-recursively enumerable.
\end{abstract}

\section{Introduction}
\label{sec:intro}


The study of constraints on entropies is a central topic of research in information theory. In fact, more than 30 years ago, Pippenger~\cite{pippenger1986} asserted that
constraints on entropies are the ``{\em laws of information theory}''
and asked whether the {\em polymatroidal axioms} form the complete laws of information theory, i.e., whether
 every constraint on entropies can be derived from the polymatroidal axioms.
These axioms consist of the following three types of constraints:
 (1) $H(\emptyset) = 0$, (2) $H(X) \leq H(X \cup Y)$ (monotonicity),
and (3) $H(X) + H(Y) \geq H(X \cap Y) + H(X \cup Y)$ (submodularity).
It is known that the polymatroidal axioms are equivalent to Shannon's basic inequalities, that is, to the non-negativity of the entropy, conditional entropy, mutual information, and conditional mutual information~\cite{yeung2012book}.
In a celebrated result published in
1998, Zhang and Yeung~\cite{zhang1998characterization} answered Pippenger's
question negatively by finding a linear inequality that is satisfied by all entropic functions, but cannot be derived from the polymatroidal axioms.

Zhang and Yeung's result became the catalyst for the discovery of
other information laws that are not captured by the polymatroidal
axioms (e.g.,
\cite{DBLP:journals/tit/KacedR13,DBLP:conf/isit/Matus07}).  In
particular, we now know that there are more elaborate laws, such as
conditional inequalities, or inequalities expressed using $\max$,
which find equally important applications in a variety of areas. For
example,  implications between conditional independence statements
of discrete random variables can be expressed as conditional
information inequalities.  In another example, we have recently shown
that conjunctive query containment under bag semantics is at least as
hard as checking information inequalities using
$\max$~\cite{DBLP:journals/corr/abs-1906-09727}.  Despite the
extensive research on various kinds of information inequalities, to
the best of our knowledge nothing is known about the algorithmic
aspects of the associated decision problem: check whether a given
information law is valid.

In this paper, we initiate a study of algorithmic problems that arise
naturally in information theory, and establish several results.  To this
effect, we introduce a generalized form of information inequalities,
which we call {\em Boolean information constraints}, consisting of
Boolean combinations of linear information inequalities, and define
their associated decision problems. Since it is still an open problem
whether linear information inequalities, which are the simplest kind
of information laws, are decidable,  we focus on placing
these decision problems in the arithmetical hierarchy, also known as
the Kleene-Mostowski hierarchy \cite{rogers1967theory}.  The
arithmetical hierarchy has been studied by mathematical logicians since
the late 1940s; moreover, it directly influenced the
introduction and study of the polynomial-time hierarchy by Stockmeyer
\cite{stockmeyer1976polynomial}. The first level of the arithmetical hierarchy consists
of the collection $\Sigma^0_1$ of all recursively enumerable sets and
the collection $\Pi^0_1$ of the complements of all recursively
enumerable sets. The higher levels $\Sigma^0_n$ and $\Pi^0_n$,
$n\geq 2$, are defined using existential and universal quantification
over lower levels.   
%
%
We prove a number of results, including the following.

\begin{itemize}
\item[(1)] Checking the validity of a Boolean information constraint
  arising from a monotone Boolean formula (in particular, a
  max information inequality) is in $\Pi^0_1$
  (Theorem~\ref{th:co:re}).
\item[(2)] Checking the validity of a conditional information
  inequality whose antecedents are ``tight'' is in $\Pi^0_3$
  (Corollary~\ref{cor:lambda:1:tight}).  ``Tight'' inequalities are
  defined in Section~\ref{sec:beyond}, and include conditional
  independence assertions between random variables.
\item[(3)] Checking the validity of a conditional information
  inequality whose antecedents have ``slack'' and are group-balanced
  is in $\Sigma^0_2$ (Corollary~\ref{cor:lambda:1:slack}).
\item[(4)] Checking the validity of a group-balanced, max information
  inequality is Turing equivalent to checking the validity of an
  information inequality (Corollary~\ref{cor:maxiip:iip}).
\end{itemize}


While the decidability of linear information inequalities (the
simplest kind considered in this paper) remains open, a separate
important question is whether more complex Boolean information
constraints are any harder.  For example, some conditional
inequalities, or some $\max$-inequalities can be proven from a simple
linear inequality, hence they do not appear to be any harder.  However, Kaced
and Romashchenko~\cite{DBLP:journals/tit/KacedR13} proved that there
exist conditional inequalities that are {\em essentially conditional}, which means 
 that they do not follow from a linear inequality.  (We give an
example in Equation~\eqref{eq:k:s}.)  We prove here that any conditional
information inequality with slack is essentially {\em unconditioned}
(Corollary~\ref{cor:lambda:many}; see also
Equation(\ref{eq:lambda:many:consequent:slack})), and that any
$\max$-inequality also follows from a single linear inequality
(Theorem~\ref{th:lambdas:real}).

A subtle complication involving these results is whether by
``validity'' it is meant that the given Boolean information constraint
holds for the set of all entropic vectors over $n$ variables, denoted by
$\entropic$, or for its topological closure, denoted by $\aentropic$.
It is well known that these two spaces differ for all $n\geq 3$.  With
the exception of (1) above, which holds for both $\entropic$ and
$\aentropic$, our results are only for $\aentropic$.  A problem of
special interest is the implication between conditional independence
statements of discrete random variables, and this amounts to
checking the $\entropic$-validity of a tight conditional information
inequality; it is known that this problem is not finitely
axiomatizable~\cite{StudenyCINoCharacterization1990}, and its
decidability remains open.  Our result (2) above does not apply here
because it is a statement about $\aentropic$-validity.  However, we
prove  that the implication problem for conditional independence
statements is in $\Pi^0_1$ (Theorem~\ref{th:ci:in:pi1}).


\section{Background and Notations}

Throughout this paper, vectors and tuples are denoted by bold-faced
letters, and random variables are capitalized. We write
$\bm x \cdot \bm y \defeq \sum_i x_iy_i$ for the dot product of
$\bm x, \bm y \in \R^m$.  For a given set $S \subseteq \R^m$,
$S$ is {\em convex} if $\bm x, \bm y \in S$ and $\theta \in [0,1]$
implies $\theta \bm x + (1-\theta)\bm y \in S$; $S$ is called a {\em
  cone} if $\mv x \in S$ and $\theta \geq 0$ implies
$\theta \mv x \in S$; the topological closure of $S$ is denoted by
$\cl{S}$; and, finally,
$S^* \defeq \setof{\bm y}{\forall \bm x \in S, \bm x \cdot \bm y \geq
  0}$ denotes the {\em dual cone} of $S$.  It is known that $S^*$ is
always a closed, convex cone.  
We provide more background in \conferenceorfull{the full version~\cite{arxiv-version}}{Appendix~\ref{appendix:cones}}.


For a random variable $X$ with a fixed finite domain $D$ and a probability mass
function (pmf) $p$, its (binary) {\em entropy} is defined by
\begin{align}
    H(X) &\defeq  - \sum_{x \in D} p(x) \cdot \log p(x) \label{eq:h:sum:p}
\end{align}
In this paper all logarithms are in base $2$.  



Fix a joint distribution over $n$ finite random variables
$\bm V \defeq \{X_1,\dots,X_n\}$.  For each $\alpha \subseteq [n]$,
let $\bm X_\alpha$ denote the random (vector-valued) variable
$(X_i : i \in \alpha)$.  Define the set function
$h : 2^{[n]} \rightarrow \R_+$ by setting
$h(\alpha) \defeq H(\bm X_\alpha)$, for all $\alpha \subseteq [n]$.
With some abuse, we blur the distinction between the set $[n]$ and the
set of variables $\bm V = \set{X_1, \ldots, X_n}$, and write
$H(\bm X_\alpha)$, $h(\bm X_\alpha)$, or $h(\alpha)$ interchangeably.
We call the function $h$ an {\em entropic function}, and also identify
it with a vector
$\bm h \defeq  (h(\alpha))_{\alpha \subseteq [n]} \in \R^{2^n}_+$, 
which is called an {\em entropic vector}.  Note that most texts and
papers on this topic drop the component $h(\emptyset)$, which is
always 0, leading to entropic vectors in $\R^{2^n-1}$.  We prefer to
keep the $\emptyset$-coordinate to simplify notations.  The implicit
assumption $h(\emptyset)=0$ is used through the rest of the paper.

The set of entropic functions/vectors is denoted by
$\entropic \subseteq \R^{2^n}_+$.  Its topological closure, denoted
by $\aentropic$, is the set of {\em almost entropic}
vectors (or functions).  It is known~\cite{yeung2012book} that
$\entropic \subsetneq \aentropic$ for $n \geq 3$.  In general,
$\entropic$ is neither a cone nor convex, but its topological closure
$\aentropic$ is a closed convex cone~\cite{yeung2012book}.

Every entropic function $h$ satisfies the following {\em basic Shannon
  inequalities}:
\begin{align*}
    h(\bm Y \cup \bm X) \geq h(\bm X) && h(\bm X) + h(\bm Y) \geq h(\bm X \cup \bm Y) + h(\bm X \cap \bm Y)
\end{align*}
called {\em monotonicity} and {\em submodularity} respectively.  Any
inequality obtained by taking a positive linear combination of Shannon
inequalities is called a {\em Shannon-type inequality}.

Throughout this paper we will abbreviate the union $\bm X \cup \bm Y$
of two sets of variables as $\bm X \bm Y$.  The quantities
$h(\bm Y | \bm X) \defeq  h(\bm X \bm Y) - h(\bm X)$ and 
$I_h(\bm Y; \bm Z | \bm X) \defeq  h(\bm X \bm Y) + h(\bm X \bm Z) -
h(\bm X \bm Y \bm Z) - h(\bm X)$
are called the {\em conditional entropy} and the {\em conditional
  mutual information} respectively.  It can be easily checked that
$h(\bm Y | \bm X) \geq 0$ and $I_h(\bm Y; \bm Z|\bm X) \geq 0$ are
Shannon-type inequalities.

\begin{remark}
  The established notation
  $\entropic$~\cite{Yeung:2008:ITN:1457455,DBLP:journals/tit/ZhangY97,DBLP:conf/isit/Chan07}
  for the set of entropic vectors is unfortunate, because the star in
  this context does {\bf not} represent the dual cone.  We will
  continue to denote by $\entropic$ the set of entropic vectors (which
  is not a cone!), and use explicit parentheses, as in
  $(\entropic)^*$, to represent the dual cone.
\end{remark}

\section{Boolean information Constraints}

Most of this paper considers the following problem: given a Boolean
combination of information inequalities, check whether it is valid.
However in
  Section~\ref{sec:recognizability} we briefly discuss the dual problem,
namely, recognizing whether a given vector $\bm h$ is an entropic
vector (or an almost entropic vector).

A {\em Boolean function} is a function
$F : \set{0,1}^m \rightarrow \set{0,1}$.  We often denote its inputs
with variables $Z_1, \ldots, Z_m \in \set{0,1}$, and write
$F(Z_1,\ldots, Z_m)$ for the value of the Boolean function.

\subsection{Problem Definition}

A vector $\bm c \in \R^{2^n}$ defines the following (linear) {\em information inequality}:
\begin{align}
    \bm c \cdot \bm h = \sum_{\alpha \subseteq [n]} c_\alpha h(X_\alpha) & \geq 0 \label{eq:iip}.
\end{align}
The information inequality is said to be {\em valid} if it holds for
all vectors $\bm h \in \entropic$; equivalently, $\bm c$ is in the
dual cone, $\bm c \in (\entropic)^*$.  By continuity, an information
inequality holds $\forall \bm h \in \entropic$ iff it holds
$\forall \bm h \in \aentropic$. In 1986,
Pippenger~\cite{pippenger1986} defined the ``{\em laws of information
  theory}'' as the set of all information inequalities, and asked
whether all of them are Shannon-type inequalities.  This was answered
negatively by Zhang and Yeung in
1998~\cite{zhang1998characterization}.  We know today that several
applications require more elaborate laws, such as max-inequalities and
conditional inequalities. Inspired by these new laws, we define the
following generalization.

\begin{definition}
  To each Boolean function $F$ with $m$ inputs, and every $m$ vectors
  ${\bm c}_j \in \R^{2^n}, j \in [m]$, we associate the following {\em
    Boolean information constraint}:
\begin{align}
  F({\bm c}_1 \cdot \bm h\geq 0, \ldots, {\bm c}_m \cdot \bm h \geq 0).
\label{eq:entropic:gip}
\end{align}
\end{definition}

For a set $S \subseteq \R^{2^n}$, a Boolean information constraint is
said to be {\em $S$-valid} if it holds for all $\bm h \in S$.  Thus,
we will distinguish between $\entropic$-validity and
$\aentropic$-validity.
Unlike in the case of information inequalities,
these two notions of validity no longer coincide for arbitrary Boolean information constraints in general, as we
explain in what follows.

\begin{definition}
  Let $F$ be a Boolean function.  The {\em entropic Boolean
    information constraint} problem parameterized by $F$, denoted
  by $\ebic(F)$, is the following: given $m$ integer vectors
  ${\bm c}_j \in \Z^{2^n}$, where $j \in [m]$, check whether the constraint
  \eqref{eq:entropic:gip} holds for all entropic functions
  $\bm h \in \entropic$.  In the {\em almost-entropic} version,
  denoted by $\aebic(F)$, we replace $\entropic$ by $\aentropic$.
\end{definition}

The inputs ${\bm c}_j, j \in [m]$, to these problems are required to
be integer vectors in order for $\ebic(F)$ and $\aebic(F)$ to be
meaningful computational problems.  Equivalently, one can require the
inputs to be rational vectors ${\bm c}_j \in \Q^{2^n}, j \in [m]$.

Let $F$ be a Boolean function.  $F$ can be written as a conjunction of
clauses $F = C_1 \wedge C_2 \wedge \cdots$, where each clause is a
disjunction of literals.  Equivalently, a clause $C$ has this form:
\begin{align}
  & (Z_1' \wedge \cdots \wedge Z_k') \Rightarrow (Z_1 \vee \cdots \vee Z_\ell)
\label{eq:clause}
\end{align}
where $Z_1', \ldots, Z_k', Z_1, \ldots, Z_\ell$ are distinct Boolean
variables.  Checking $\ebic(F)$ is equivalent to checking $\ebic(C)$,
for each clause of $F$ (and similarly for $\aebic(F)$); therefore and
without loss of generality, we will assume in the rest of the paper
that $F$ consists of a single clause~\eqref{eq:clause} and study the
problem along these dimensions:

{\bf Conditional and Unconditional Constraints} When $k=0$ (i.e., when
the antecedent is empty), the formula $F$ is {\em monotone}, and we
call the corresponding Boolean information constraint {\em
  unconditional}.  When $k>0$, the formula $F$ is {\em non-monotone},
and we call the corresponding constraint {\em conditional}.
{\bf Simple and Max Constraints} When $k = 0$ and $\ell=1$, then we say that $F$
defines a {\em simple} inequality; when $k = 0$ and $\ell>1$, then we say that $F$
defines a {\em $\max$-inequality}.  The case when $\ell=0$ and $k>0$ is not
interesting because $F$ is not valid, since the zero-vector
$\bm h = \bm 0$ violates the constraint.

\eat{
\begin{description}
    \item[Antecedent] When $k=0$, the formula $F$ is {\em monotone}, and we call the
corresponding Boolean information constraint {\em unconditional}. When $k>0$, the formula $F$ is {\em
non-monotone}, and we call the corresponding constraint {\em conditional}.
    \item[Consequent] When $m=1$, we say that $F$ defines a {\em simple
        inequality}.  When $m>1$, we say that $F$ defines a
        {\em $\max$-inequality}.  The case when $m=0$ is not interesting because $F$
  is not valid, since the zero-vector $\bm h = \bm 0$ violates the constraint.
\end{description}
}

%
%
%
%
%


\subsection{Examples and Applications}
\label{subsec:examples}

This section presents examples and applications of Boolean Function
Information Constraints and their associated decision problems.  A
summary of the notations is in Fig.~\ref{fig:notations}.

\begin{figure}
  \centering
  \begin{tabular}{|l|c|c|l|} \hline\hline
       &\multicolumn{2}{c|}{Abbreviation} &  \\ \cline{2-3}
Problem& Entropic & Almost-       &  Simple Example       \\
       &          & entropic       &         \\ \hline\hline
Boolean information & $\ebic(F)$ & $\aebic(F)$ & $h(XY) \leq \frac{2}{3}h(XYZ)\Rightarrow$ \\
constraint  & & & $\max(h(YZ),h(XZ)) \geq \frac{2}{3}h(XYZ)$ \\ \hline
Information Inequality & \multicolumn{2}{c|}{$\iip$} & $h(XY)+h(YZ)+h(XZ)\geq 2h(XYZ)$\\ \hline
Max-Information Inequality & \multicolumn{2}{c|}{$\maxiip$} & $\max(h(XY),h(YZ),h(XZ))\geq\frac{2}{3}h(XYZ)$\\ \hline
Conditional Information  & $\econdiip$ & $\aecondiip$ & $((h(XY) \leq \frac{2}{3}h(XYZ)) \wedge (h(YZ) \leq \frac{2}{3}h(XYZ)))$\\
Inequality & & & $\Rightarrow h(XZ) \geq \frac{2}{3}h(XYZ)$ \\ \hline
Conditional Independence & CI & (no name) & $(I(X;Y)=0 \wedge I(X;Z|Y)=0) \Rightarrow I(X;Z)=0$ \\ \hline\hline
  \end{tabular}
  \caption{Notations for various Boolean Information Constraint
    Problems.}
  \label{fig:notations}
\end{figure}

\subsubsection{Information Inequalities} \label{subsubsection:ii} We
start with the simplest form of a Boolean information constraint,
namely, the linear information inequality in Eq.~\eqref{eq:iip}, which
arises from the single-variable Boolean formula $Z_1$.  We will call
the corresponding decision problem the {\em information-inequality
  problem}, denoted by $\iip$: given a vector of integers $\bm c$,
check whether Eq.~\eqref{eq:iip} is $\entropic$-valid or,
equivalently, $\aentropic$-valid.  Pippenger's question from 1986 was
essentially a question about decidability.  Shannon-type inequalities
are decidable in exponential time using linear programming methods, and software packages have been developed for this purpose \cite[Chapter 13]{yeung2012book} (it is not known, however, if there is a matching lower bound in the complexity of this problem).
Thus,  if every information inequality were a
Shannon-type inequality, then information inequalities would be decidable.  However, Zhang and
Yeung's gave the first example of a non-Shannon-type information
inequality~\cite{zhang1998characterization}.  Later, Mat{\'{u}}{\v
  s}~\cite{DBLP:conf/isit/Matus07} proved that, when $n\geq 4$
variables, there exists infinitely many inequivalent non-Shannon
entropic inequalities.  More precisely, he proved that the following
is a non-Shannon inequality, for every $k \geq 1$:
\begin{align}
I_h(C;D|A) +\frac{k+3}{2}I_h(C;D|B) + I_h(A;B) + \frac{k-1}{2}I_h(B;C|D)+\frac{1}{k}I_h(B;D|C)
\geq I_h(C;D)
\label{eq:matus}
\end{align}
This ruined any hope of proving decidability of information
inequalities by listing a finite set of axioms.  To date, the study of
non-Shannon-type inequalities is an active area of
research~\cite{MR2042797,MR1958013,MR1907510}, and the question
whether $\iip$ is decidable remains open.

Hammer et al.\ \cite{hammer2000inequalities}, showed that, up to
logarithmic precision, information inequalities are equivalent to
linear inequalities in Kolmogorov complexity (see also \cite[Theorem
3.5]{grunwald2004shannon}).

\subsubsection{Max Information Inequalities} \label{subsubsec:maxiip}
Next, we consider constraints defined by a disjunction of linear
inequalities, in other words
$(\bm c_1 \cdot \bm h \geq 0) \vee \cdots \vee (\bm c_m \cdot \bm h
\geq 0)$, where $\bm c_j \in \R^{2^n}$.  This is equivalent to:
\begin{align}
\max(\bm c_1 \cdot \bm h, \bm c_2 \cdot \bm h, \ldots, \bm c_m \cdot  \bm h) \geq 0
\end{align}
and, for that reason, we call them {\em Max information inequalities}
and denote the corresponding decision problem by $\maxiip$. As before,
$\entropic$-validity and $\aentropic$-validity coincide.

{\bf Application to Constraint Satisfaction and Database Theory}
Given two finite structures
%
%
$\bm A$ and $\bm B$, we write $\hom({\bm A}, {\bm B})$ for the set of
homomorphisms from ${\bm A}$ to ${\bm B}$.  We say that ${\bm B}$ {\em
  dominates} structure ${\bm A}$, denote by ${\bm A} \preceq {\bm B}$,
if for every finite structure ${\bm C}$, we have that
$|\hom({\bm A}, {\bm C})| \leq |\hom({\bm B}, {\bm C})|$.  The {\em
  homomorphism domination problem} asks whether
${\bm A} \preceq {\bm B}$, given ${\bm A}$ and ${\bm B}$.  In database
theory this problem is known as the {\em query containment problem
  under bag semantics}~\cite{DBLP:conf/pods/ChaudhuriV93}.  In that
setting we are given two Boolean conjunctive queries $Q_1, Q_2$, which
we interpret using bag semantics, i.e., given a database $D$, the
answer $Q_1(D)$ is the number of homomorphisms
$Q_1 \rightarrow D$~\cite{DBLP:conf/pods/KolaitisV98}.  {\em $Q_1$ is
  contained in $Q_2$ under bag semantics} if $Q_1(D) \leq Q_2(D)$ for
every database $D$.  It is open whether the homomorphism domination
problem is decidable.

Kopparty and Rossman~\cite{HDE} described a $\maxiip$ problem that yields
a sufficient condition for homomorphism domination.  In recent
work~\cite{DBLP:journals/corr/abs-1906-09727} we proved that, when
$\bm B$ is acyclic, then that condition is also necessary, and,
moreover, the domination problem for acyclic $\bm B$ is
Turing-equivalent to $\maxiip$.  Hence, any result on the complexity
of $\maxiip$ immediately carries over to the homomorphism domination
problem for acyclic $\bm B$, and vice versa.

We illustrate here Kopparty and Rossman's $\maxiip$ condition on a
simple example.  Consider the following two Boolean conjunctive
queries: $Q_1() = R(u,v)\wedge R(v,w) \wedge R(w,u)$,
$Q_2() = R(x,y)\wedge R(x,z)$; interpreted using bag semantics, $Q_1$
returns the number of triangles and $Q_2$ the number of V-shaped
subgraphs.  Kopparty and Rossman proved that $Q_1 \preceq Q_2$ follows
from the following max-inequality:

{\footnotesize
\begin{align}
    \max \bigl\{ 2h(XY) - h(X) - h(XYZ), 2h(YZ) - h(Y)-h(XYZ),
    2h(XZ) - h(Z) -  h(XYZ) \bigr\} & \geq 0\label{eq:kopparty:rossman}
\end{align}
}
%

\subsubsection{Conditional Information  Inequalities} \label{subsubsection:cond:iip} A {\em conditional
  information inequality} has the form:
\begin{align}
  (\bm c_1 \cdot \bm h \geq 0 \wedge \cdots \wedge \bm c_k  \cdot \bm h \geq 0) \Rightarrow
  \bm c_0 \cdot \bm h \geq 0
\label{eq:cond:ii}
\end{align}
%
Here we need to distinguish between $\entropic$-validity and
$\aentropic$-validity, and denote by $\econdiip$ and $\aecondiip$ the
corresponding decision problems.  Notice that, without loss of
generality, we can allow equality in the antecedent, because
$\bm c_i \cdot \bm h = 0$ is equivalent to
$\bm c_i \cdot \bm h \geq 0 \wedge - \bm c_i \cdot \bm h \geq 0$.

Suppose that there exist $\lambda_1 \geq 0, \ldots, \lambda_m \geq 0$
such that the inequality
$\bm c_0 \cdot \bm h - (\sum_i \lambda_i \bm c_i \cdot \bm h) \geq 0$
is valid; then Eq.~\eqref{eq:cond:ii} is, obviously, also valid.
Kaced and Romashchenko~\cite{DBLP:journals/tit/KacedR13} called
Eq. (\ref{eq:cond:ii}) an {\em essentially conditioned inequality} if
no such $\lambda_i$'s exist, and discovered several valid conditional
inequalities that are essentially conditioned.

{\bf Application to Conditional Independence}
Fix three set of random variables $\bm X, \bm Y, \bm Z$.  A {\em
  conditional independence} (CI) statement is a statement of the form
$\phi = (\bm Y \upmodels \bm Z \suchthat \bm X)$, and it asserts that
$\bm Y$ and $\bm Z$ are independent conditioned on $\bm X$.  A {\em CI
  implication} is a statement
$\varphi_1 \wedge \cdots \wedge \varphi_k \Rightarrow \varphi_0$,
where $\varphi_i, i \in \{0,\dots,k\}$ are CI statements.  The {\em CI
  implication
  problem} 
is: given an implication, check if it is valid for all discrete
probability distributions.  Since
$(\bm Y \upmodels \bm Z \suchthat \bm X) \Leftrightarrow I_h(\bm Y;
\bm Z | \bm X) = 0 \Leftrightarrow -I_h(\bm Y; \bm Z | \bm X) \geq 0$,
the CI implication problem is a special case of $\econdiip$.

The CI implication problem has been studied extensively in the
literature~\cite{DBLP:journals/tse/Lee87,StudenyCINoCharacterization1990,GeigerPearl1993,DBLP:journals/corr/abs-1812-09987}.
Pearl and Paz~\cite{DBLP:conf/ecai/PearlP86} gave a sound, but
incomplete, set of {\em graphoid axioms},
Studen\'y~\cite{StudenyCINoCharacterization1990} proved that no finite
axiomatization exists, while Geiger and Pearl~\cite{GeigerPearl1993}
gave a complete axiomatization for two restricted classes, called
saturated, and marginal CIs.
See~\cite{DBLP:conf/uai/WaalG04,DBLP:journals/ipl/GyssensNG14,DBLP:journals/ai/NiepertGSG13}
for some recent work on the CI implication problem.  The decidability
of the CI implication problem remains open to date.

Results in~\cite{DBLP:journals/tit/KacedR13} imply that the following
CI implication is essentially conditioned
(see~\cite{DBLP:journals/corr/abs-1812-09987}):
\begin{align}
 I_h(C;D|A) = I_h(C;D|B) = I_h(A;B) = I_h(B;C|D)=0 & \Longrightarrow I_h(C;D)=0
\label{eq:k:s}
\end{align}
While a CI implication problem is an instance of an {\em entropic}
conditional inequality, one can also consider the question whether a
CI implication statement holds for all {\em almost entropic}
functions; for example the implication~\eqref{eq:k:s} holds for all
almost entropic functions.  Kaced and
Romashchenko~\cite{DBLP:journals/tit/KacedR13} proved that these two
problems differ, by giving examples of CI implications that hold for
all entropic functions but fail for almost entropic functions.


\subsubsection{Group-Theoretic Inequalities} There turns out to be a
way to ``rephrase'' $\iip$ as a decision problem in group theory; This
was a wonderful result by Chan and
Yeung~\cite{DBLP:journals/tit/ChanY02} (see
also~\cite{DBLP:conf/isit/Chan07}).  A tuple $(G; G_1, \dots, G_n)$ is
called a {\em group system} if $G$ is a finite group and
$G_1, \ldots, G_n \subseteq G$ are $n$ subgroups.  For any
$\alpha \subseteq [n]$, define
$G_\alpha := \bigcap_{i \in \alpha} G_i$; implicitly, we set
$G_\emptyset := G$.  A vector $\bm c \subseteq \R^{2^n}$ defines the
following {\em group-theoretic inequality}:
\begin{align}
    \sum_{\alpha \subseteq [n]} c_\alpha \log \frac{|G|}{|G_\alpha|} \geq 0
\label{eq:group:inequality}
\end{align}
\begin{theorem}[\cite{DBLP:journals/tit/ChanY02}] An information inequality
  \eqref{eq:iip} is $\entropic$-valid if and only if the corresponding group-theoretic
  inequality \eqref{eq:group:inequality} holds for all group systems $(G, G_1, \dots, G_n)$,
\end{theorem}

In particular, a positive or negative answer to the decidability
problem for $\iip$ immediately carries over to the validity problem of
group-theoretic inequalities of the form~\eqref{eq:group:inequality}.
We note that the group-theoretic inequalities considered here are
different from the word problems in group, see e.g. the
survey~\cite{miller-decision-problems-for-groups}; the undecidability
results for word problems in groups do not carry over to the
group-theoretic inequalities and, thus, to information inequalities.

\subsubsection{Application to Relational Query Evaluation}
The problem of bounding the number of copies of a graph inside of
another graph has a long and interesting
history~\cite{MR1639767,MR599482,MR859293,DBLP:conf/pods/000118}.  The subgraph
homomorphism problem is a special case of the relational query
evaluation problem, in which case we want to find an upper bound on
the output size of a full conjunctive query. Using the entropy
argument from~\cite{MR859293}, {\em Shearer's lemma} in particular,
Atserias, Grohe, and Marx~\cite{DBLP:journals/siamcomp/AtseriasGM13}
established a tight upper bound on the answer to a full conjunctive
query over a database.  Note that Shearer's lemma is a Shannon-type
inequality.
Their result was extended to include functional dependencies and more
generally degree constraints in a series of recent work in database
theory~\cite{DBLP:journals/jacm/GottlobLVV12,DBLP:conf/pods/KhamisNS16,DBLP:conf/pods/Khamis0S17}.
All these results can be cast as applications of Shannon-type
inequalities.  For a simple example, let $R(X,Y), S(Y,Z), T(Z,U)$ be
three binary relations (tables), each with $N$ tuples, then their join
$R(X,Y) \Join S(Y,Z) \Join T(Z,U)$ can be as large as $N^2$ tuples.
However, if we further know that the functional dependencies
$XZ \rightarrow U$ and $YU \rightarrow X$ hold in the output, then one
can prove that the output size is $\leq N^{3/2}$, by using the
following Shannon-type information inequality:
\begin{align}
    h(XY)+h(YZ) + h(ZU) + h(X|YU) + h(U|XZ) & \geq 2 h(XYZU)
\end{align}
%
While the tight upper bound of any conjunctive query can be proven using
only Shannon-type inequalities, this no longer holds when the
relations used in the query are constrained to satisfy functional
dependencies.  In that case, the tight upper bound can always be
obtained from an information inequality, but Abo Khamis et
al.~\cite{DBLP:conf/pods/Khamis0S17} gave an example of a conjunctive
query for which the tight upper bound requires a non-Shannon
inequality.

\subsubsection{Application to Secret Sharing}
\label{subsec:secret:sharing}
An interesting application of conditional information inequalities is
secret sharing, which is a classic problem in cryptography,
independently introduced by Shamir~\cite{DBLP:journals/cacm/Shamir79}
and Blakley~\cite{blakley}. The setup is as follows. There is a set
$P$ of {\em participants}, a {\em dealer} $d \notin P$, and an {\em
  access structure} $\calF \subset 2^P$. The access structure is
closed under taking superset: $A \in \calF$ and $A \subseteq B$
implies $B \in \calF$.  The dealer has a secret $s$, from some finite
set $K$, which she would like to share in such a way that every set
$F \in \calF$ of participants can recover the secret $s$, but every
set $F \notin \calF$ knows {\em nothing} about $s$.  The dealer shares
her secret by using a {\em secret sharing scheme}, in which she gives
each participant $p \in P$ a {\em share} $s_p \in K_p$, where $K_p$ is
some finite domain.  The scheme is designed in such a way that from
the tuple $(s_p)_{p \in F}$ one can recover $s$ if $F \in \calF$, and
conversely one cannot infer any information about $s$ if
$F \notin \calF$.

One way to formalize secret sharing uses information theory (for other
formalisms, see~\cite{DBLP:conf/codcry/Beimel11}).  We identify the
participants $P$ with the set $[n-1]$, and the dealer with the number
$n$. A secret sharing scheme on $P$ with access structure
$\calF \subseteq 2^P$ is a joint distribution on $n$ discrete random
variables $(X_1, \dots, X_n)$ satisfying:
\begin{itemize}
    \item[(i)] $H(X_n) > 0$
    \item[(ii)] $H(X_n \suchthat \bm X_F) = 0$ if $F \in \calF$
    \item[(iii)] $H(X_n \suchthat \bm X_F) = H(X_n)$ if
      $F \notin \calF$; equivalently, $I_H(X_n; \bm X_F) = 0$.
\end{itemize}
Intuitively, $X_i$ denotes the share given to the $i$th participant,
and $X_n$ is the unknown secret. It can be shown, without loss of
generality, that $(i)$ can be replaced by the assumption that the
marginal distribution on $X_n$ is
uniform~\cite{DBLP:journals/ipl/BlundoSV98}, which encodes the fact
that the scheme does not reveal any information about the secret
$X_n$.  Condition $(ii)$ means one can recover the secret from the
shares of qualified participants, while condition $(iii)$ guarantees
the complete opposite.  A key challenge in designing a good secret
sharing scheme is to reduce the total size of the shares. The only
known~\cite{DBLP:journals/joc/Csirmaz97,DBLP:journals/joc/CapocelliSGV93,DBLP:journals/tit/KarninGH83}
way to prove a {\em lower bound} on share sizes is to lower bound the
{\em information ratio} $\frac{\max_{p \in P} H(X_p)}{H(X_n)}$.  In
order to prove that some number $\ell$ is a lower bound on the
information ratio, we need to check that
$\max_{i \in [n-1]} \bigl\{ h(X_i) - \ell \cdot h(X_n) \bigr\} \geq 0$
holds for all entropic functions $\bm h \in \entropic$ satisfying the
extra conditions (i), (ii), and (iii) above.  Equivalently, $\ell$ is
a lower bound on the information ratio if and only if the following Boolean
information constraint is $\entropic$-valid:
\begin{align*}
    \bigwedge_{F\in \calF} (h(X_n \suchthat \bm X_F) = 0)
    \wedge
    \bigwedge_{F\not\in \calF} (I_h(X_n; \bm X_F) = 0)
    \Longrightarrow
    (h(X_n) = 0) \vee
    \bigl[
        \bigvee_{i \in [n-1]} (h(X_i) \geq \ell \cdot h(X_n))
    \bigr]
\end{align*}


\section{Placing $\ebic$ and $\aebic$ in the Arithmetical Hierarchy}

What is the complexity of $\ebic(F)$ / $\aebic(F)$?  Is it even
decidable?  As we have seen there are numerous applications of the
Boolean Information Constraint problem, hence any positive or negative
answer, even for special cases, would shed light on these
applications.  While their (un)decidability is currently open, in this
paper we provide several upper bounds on their complexity, by placing
them in the arithmetical hierarchy.


We briefly review some concepts from computability theory.  In this
setting it is standard to assume objects are encoded as natural
numbers.  A set $A \subseteq \N^k$, for $k \geq 1$, is \emph{Turing
  computable}, or \emph{decidable}, if there exists a Turing machine
that, given $x \in \N^k$ decides whether $x \in A$.  A set $A$ is
\emph{Turing reducible} to $B$ if there exists a Turing machine with
an oracle for $B$ that can decide membership in $A$.  The {\em
  arithmetical hierarchy} consists of the classes of sets $\Sigma^0_n$
and $\Pi^0_n$ defined as follows. The class $\Sigma^0_n$ consists of
all sets of the form
$\set{x \suchthat \exists y_1 \forall y_2 \exists y_3 \cdots \text{\sf
    Q} y_n R(x, y_1, \dots, y_n)}$, where $R$ is an $(n+1)$-ary
decidable predicate, \text{\sf Q} $ = \exists$ if $n$ is odd, and
$Q= \forall$ if $n$ is even. In a dual manner, the class $\Pi^0_n$
consists of sets of the form
$\set{x \suchthat \forall y_1 \exists y_2 \forall y_3 \cdots \text{\sf
    Q} y_n R(x, y_1, \dots, y_n)}$.
%
Then $\Sigma^0_0 = \Pi^0_0$ are the decidable sets, while $\Sigma^0_1$
consists of the {\em recursively enumerable} sets, and $\Pi^0_1$
consists of the {\em co-recursively enumerable} sets.  It is known
that these classes are closed under union and intersection, but not
under complements, and that they form a strict hierarchy,
$\Sigma^0_n, \Pi^0_n \subsetneq (\Sigma^0_{n+1} \cap \Pi^0_{n+1})$.
For more background, we refer to~\cite{rogers1967theory}.
Our goal is to place the problems $\ebic(F)$, $\aebic(F)$, and their
variants in concrete levels of the arithmetical hierarchy.

\subsection{Unconditional Boolean Information Constraints}

\label{sec:unconditional:bic}

We start by discussing unconditional Boolean information constraints,
or, equivalently, a Boolean information constraint defined by a
monotone Boolean formula $F$. The results here are rather simple; we
include them only as a warmup for the less obvious results in later
sections.  Based on our discussion in Sections~\ref{subsubsection:ii} and
~\ref{subsubsec:maxiip}, we have the following result.

\begin{theorem} \label{th:entropic:to:almost:entropic} If $F$ is monotone,
  then $\ebic(F)$ and $\aebic(F)$ are equivalent problems.
\end{theorem}
%

Next, we prove that these problems are co-recursively enumerable, by using the
following folklore fact.  A {\em representable set of $n$ random
  variables} is a finite relation $\Omega$ with $N$ rows and $n+1$
columns $X_1, \ldots, X_n, p$, where column $p$ contains rational
probabilities in $[0,1] \cap \Q$ that sum to 1.  Thus, $\Omega$
defines $n$ random variables with finite domain and probability mass
given by rational numbers.  We denote $\bm h^\Omega$ its entropic
vector.  By continuity of Eq.(\ref{eq:h:sum:p}), we obtain:

\begin{proposition} \label{prop:rational:p} For every entropic vector
  $\bm h \in \entropic$ and every $\varepsilon > 0$, there exists
  a representable space $\Omega$ such that
  $\norm{\bm h - \bm h^\Omega} < \varepsilon$.
\end{proposition}

The group-characterization proven by Chan and
Yeung~\cite{DBLP:journals/tit/ChanY02} implies a much stronger version
of the proposition; we do not need that stronger version in this
paper.

\begin{theorem} \label{th:co:re} Let $F$ be a monotone Boolean formula.
  Then $\ebic(F)$ (and, hence,  $\aebic(F)$) is in $\Pi^0_1$, i.e., it is co-recursively
  enumerable.
\end{theorem}
\begin{proof} Fix $F = Z_1 \vee \cdots \vee Z_m$ and
$\bm c_i \in \Z^{2^n}$, $i \in [m]$.  We need to check:
  \begin{align}
      \forall \bm h \in \entropic: && \bm c_1 \cdot \bm h \geq 0 \vee \cdots \vee  \bm c_m \cdot \bm h\geq 0
\label{eq:ii:unrestricted}
  \end{align}
  We claim that~\eqref{eq:ii:unrestricted} is equivalent to:
  \begin{align}
    \forall \Omega && \bm c_1 \cdot \bm h^\Omega \geq 0 \vee \cdots \vee \bm c_m \cdot \bm   h^\Omega \geq 0
\label{eq:ii:representable}
  \end{align}
  Obviously~\eqref{eq:ii:unrestricted}
  implies~\eqref{eq:ii:representable}, and the opposite follows from
  Prop.~\ref{prop:rational:p}: if~\eqref{eq:ii:unrestricted} fails on
  some entropic vector $\bm h$, then it also fails on some
  representable $\bm h^\Omega$ close enough to $\bm h$.  Finally,
  ~\eqref{eq:ii:representable} is in $\Pi^0_1$ because, the property
  after $\forall\Omega$ is decidable, by expanding the definition of
  entropy~\eqref{eq:h:sum:p} in each condition
  $\bm c_i \cdot \bm h^\Omega \geq 0$, and writing the latter as
  $\sum_j a_j \log b_j \geq 0$, or, equivalently,
  $\prod_j (b_j)^{a_j} \geq 1$, where $a_j, b_j$ are rational numbers,
  which is decidable.
\end{proof}

\subsection{Conditional Boolean Information Constraints}

We now consider non-monotone Boolean functions, in other words,
conditional information constraints~\eqref{eq:cond:ii}. Since
$\entropic$- and $\aentropic$-validity no longer coincide, we study
$\ebic(F)$ and $\aebic(F)$ separately.  The results here are
non-trivial, and some proofs are
\conferenceorfull{deferred to~\cite{arxiv-version}.}
{included in the Appendix.}

\subsubsection{The Entropic Case}
\label{subsec:entropic:functions}

Our result for $\ebic(F)$ is restricted to the CI implication problem.
Recall from Sec.~\ref{subsubsection:cond:iip} that this problem
consists of checking whether an implication between statements of the
form $(\bm Y \upmodels \bm Z \suchthat \bm X)$ holds for all random
variables with finite domain, and this is equivalent to checking
whether a certain conditional inequality holds for all entropic
functions.  We prove that this problem is in $\Pi^0_1$ by using
Tarski's theorem of the decidability of the theory of reals with
$+, *$~\cite{tarski1998decision}.

\begin{theorem} \label{th:ci:in:pi1} The CI implication problem
  (Section~\ref{subsubsection:cond:iip}) is in $\Pi^0_1$.
\end{theorem}

\begin{proof} 
  Tarski has proven that the theory of reals with $+, *$ is decidable.
  More precisely, given a formula $\Phi$ in FO with symbols $+$ and
  $*$, it is decidable whether that formula is true in the model of
  real numbers $(\R, +, *)$; for example, it is decidable
  whether\footnote{$3y$ is a shorthand for $y+y+y$ and $x \geq y$ is a
    shorthand for $\exists u (x = y+u^2)$.}
  $\Phi \equiv \forall x \exists y \forall z (x^2 + 3 y \geq z \wedge
  (y^3 + yz \leq xy^2))$ is true.  We will write
  $(\R,+,*) \models \Phi$ to denote the fact that $\Phi$ is true in
  the model of reals.

  Consider a conditional inequality over a set of $n$ joint random
  variables:
  \begin{align*}
    I_h(Y_1;Z_1|X_1)=0 \wedge \cdots \wedge I_h(Y_k;Z_k|X_k)=0
    \Rightarrow  I_h(Y;Z|X) = 0
  \end{align*}

  The following algorithm returns {\em false} if the inequality fails
  on some entropic function $h$, and runs forever if the inequality
  holds for all $h$, proving that the problem is in $\Pi^0_1$:

  \begin{itemize}
  \item Iterate over all $N \geq 0$.  For each $N$, do the following
    steps.
  \item Consider $n$ joint random variables $X_1, \ldots, X_n$ where
    each has outcomes in the domain $[N]$; thus there are $N^n$
    possible outcomes.  Let $p_1, \ldots, p_{N^n}$ be real variables
    representing the probabilities of these outcomes.
  \item Construct a formula $\Delta$ stating ``there exist
    probabilities $p_1, \ldots, p_{N^n}$ for these outcomes, whose
    entropy fails the conditional inequality''.  More precisely, the
    formula consists of the following:
    \begin{itemize}
    \item Convert  each conditional independence statement in the
      antecedent   $I_h(Y_i;Z_i|X_i)=0$ into its equivalent statement
      on probabilities: $p(X_iY_iZ_i)p(X_i) = p(X_iY_i)p(X_iZ_i)$.
    \item Replace each such statement with a conjunction of statements
      of the form
      $p(X_i=x, Y_i=y, Z_i=z)\cdot p(X_i=x) = p(X_i=x, Y_i=y)\cdot p(X_i=x,
      Z_i=z)$, for all combinations of values $x,y,z$.  If
      $X_i, Y_i, Z_i$ have in total $k$ random variables, then there
      are $N^k$ combinations of values $x,y,z$, thus we create a
      conjunction of $N^k$ equality statements.
    \item Each marginal probability is a sum of atomic probabilities,
      for example $p(X_i=x, Y_i=y) = p_{k_1}+p_{k_2}+\cdots$ where
      $p_{k_1}, p_{k_2},\ldots$ are the probabilities of all outcomes
      that have $X_i=x$ and $Y_i=y$.  Thus, the equality statement in
      the previous step becomes the following formula:
      $(p_{i_1} + p_{i_2} + \cdots)(p_{j_1}+p_{j_2}+\cdots) =
      (p_{k_1}+p_{k_2}+\cdots)(p_{\ell_1}+p_{\ell_2}+\cdots)$.  There
      is one such formula for every combination of values $x,y,z$;
      denote $\Phi_i$ the conjunction of all these formulas.  Thus,
      $\Phi_i$ asserts $I_h(Y_i;Z_i|X_i) = 0$.
    \item Let $\Phi = \Phi_1 \wedge \cdots \wedge \Phi_k$.  Let $\Psi$
      be the similar formula for the consequent: thus, $\Psi$ asserts
      $I_h(Y;Z|X)=0$.
    \item Finally, construct the formula
      $\Delta \defeq \exists p_1, \ldots, \exists p_{N^n}, (\Phi
      \wedge \neg \Psi)$.
    \end{itemize}
  \item Check whether $(\R, +, *) \models \Delta$.  By Tarski's
    theorem this step is decidable.
  \item If $\Delta$ is true, then return {\em false}; otherwise,
    continue with $N+1$.
  \end{itemize}
\end{proof}

\paragraph*{Tarski's exponential function problem} One may attempt to
extend the proof above from the CI implication problem to arbitrary
conditional inequalities~\eqref{eq:cond:ii}.  To check if a
conditional inequality is valid for all entropic functions, we can
repeat the argument above: iterate over all domain sizes
$N=1,2,3,\ldots$, and check if there exists probabilities
$p_1, \ldots, p_{N^n}$ that falsify the implication
$(\bm c_1 \cdot \bm h \geq 0 \wedge \cdots \wedge \bm c_k \cdot \bm h
\geq 0) \Rightarrow \bm c_0 \cdot \bm h \geq 0$.  The problem is that
in order to express $\bm c_i \cdot \bm h \geq 0$ we need to express
the vector $\bm h$ in terms of the probabilities
$p_1, \ldots, p_{N^n}$.  To apply directly the definition of entropy
in~\eqref{eq:h:sum:p} we need to use the $\log$ function, or,
alternatively, the exponential function, and this takes us outside
the scope of Tarski's theorem.  A major open problem in model theory,
originally formulated also by Tarski, is whether decidability
continues to hold if we augment the structure of the real numbers with
the exponential function (see,
e.g.,~\cite{david:marker:exponentiation} for a discussion).
Decidability of the first-order theory of the reals with
exponentiation would easily imply that the entropic conditional
information inequality problem \econdiip\ (not just the entropic
conditional independence (CI) implication problem) is in $\Pi^0_1$,
because every condition $\bm c \cdot \bm h \geq 0$ can be expressed
using $+, *$ and the exponential function, by simply expanding the
definition of entropy in Equation \eqref{eq:h:sum:p}.

\subsubsection{The Almost-Entropic Case}
\label{sec:beyond}


Suppose the antecedent of~\eqref{eq:cond:ii} includes the condition
$\bm c \cdot \bm h \geq 0$.  Call $\bm c \in \R^{2^n}$ {\em tight} if
$\bm c \cdot \bm h \leq 0$ is $\aentropic$-valid.  When $\bm c$ is
tight, we can rewrite $\bm c \cdot \bm h \geq 0$ as
$\bm c \cdot \bm h = 0$.  If $\bm c$ is not tight, then there exists
$\bm h \in \aentropic$ such that $\bm c \cdot \bm h > 0$; in that case
we say that $\bm c$ {\em has slack}.  For example, all conditions
occurring in CI implications are tight, because they are of the form
$-I_h(Y;Z|X) \geq 0$, and more conveniently written $I_h(Y;Z|X)=0$,
while a condition like $3h(X) - 4h(YZ) \geq 0$ has slack.  We extend
the definition of slack to a set.  We say that the set
$\set{\bm c_1, \ldots, \bm c_k} \subset \R^{2^n}$ has slack if there
exists $\bm h \in \aentropic$ such that $\bm c_i \cdot \bm h > 0$ for
all $i=1,k$; notice that this is more restricted than requiring each
of $\bm c_i$ to have slack.  We present below results on the
complexity of $\aebic(F)$ in two special cases: when all antecedents
are tight, and when the set of antecedents has slack.  Both results
use the following theorem, which allows us to move one condition
$\bm c_k \cdot \bm h \geq 0$ from the antecedent to the consequent:

\begin{theorem}
  \label{th:lambda:1}
  The following statements are equivalent:
  \begin{align}
      \forall \bm h \in \aentropic: && \bigwedge_{i \in [k]} \bm c_i \cdot \bm h \geq 0 \Rightarrow &\bm c \cdot \bm h \geq 0\label{eq:lambda:1:antecedent}\\
    \forall \varepsilon>0,\exists \lambda\geq 0, \forall \bm h \in \aentropic: &&
    \bigwedge_{i \in [k-1]} \bm c_i \cdot \bm h \geq 0
    \Rightarrow  & \bm c \cdot \bm h+ \varepsilon h([n]) \geq \lambda \bm c_k \cdot \bm h\label{eq:lambda:1:consequent}
  \end{align}
  Moreover, if the set $\set{\bm c_1, \ldots, \bm c_k}$ has slack,
  then one can set $\varepsilon=0$ in
  Eq.\eqref{eq:lambda:1:consequent}.
\end{theorem}

\begin{proof} We prove here only the implication from
  \eqref{eq:lambda:1:consequent} to \eqref{eq:lambda:1:antecedent};
  the other direction is non-trivial and is proven in
  \conferenceorfull{the full version~\cite{arxiv-version}}{Appendix~\ref{app:proof:th:lambda:eps}} using only the properties of closed convex cones.  Assume condition
  \eqref{eq:lambda:1:consequent} holds, and consider any
  $\bm h \in \aentropic$ s.t.
  $\bigwedge_{i\in [k]} \bm c_i \cdot \bm h \geq 0$.  We prove that
  $\bm c \cdot \bm h \geq 0$.  For any $\varepsilon > 0$, condition
  \eqref{eq:lambda:1:consequent} states that there exists $\lambda >0$
  such that
  $\bm c \cdot \bm h+ \varepsilon h([n]) \geq \lambda \bm c_k \cdot
  \bm h$ and therefore $\bm c \cdot \bm h+ \varepsilon h([n]) \geq 0$.
  Since $\varepsilon > 0$ is arbitrary, we conclude that
  $\bm c \cdot \bm h \geq 0$, as required.
\end{proof}

By applying the theorem repeatedly, we can move all antecedents to the
consequent:

\begin{corollary}
  \label{cor:lambda:many}
  Condition~\eqref{eq:lambda:1:antecedent} is  equivalent to:
  \begin{align}
    \forall \varepsilon>0,\exists \lambda_1 \geq 0, \cdots, \exists \lambda_k \geq 0, \forall \bm h \in \aentropic: &&
    \bm c \cdot \bm h+ \varepsilon h([n]) \geq \sum_{i \in [k]} \lambda_i \bm c_i \cdot \bm h\label{eq:cor:lambda:1:consequent}
  \end{align}
  Moreover, if the set $\set{\bm c_1, \ldots, \bm c_k}$ has slack,
  then one can set $\varepsilon=0$ in
  Eq.\eqref{eq:cor:lambda:1:consequent}.
\end{corollary}

{\bf Antecedents Are Tight} We consider now the case when all
antecedents are tight, a condition that can be verified in $\Pi^0_1$,
by Th.\ref{th:co:re}.  In that case,
condition~\eqref{eq:lambda:1:antecedent} is equivalent to:
  \begin{align}
    \forall p \in \N,\exists q\in\N, \forall \bm h \in \aentropic: && \bm c \cdot \bm h+ \frac{1}{p} h([n]) \geq q \sum_{i\in[k]}\bm c_i \cdot \bm h\label{eq:lambda:many:consequent:tight}
  \end{align}
  Indeed, the non-trivial direction
  \eqref{eq:cor:lambda:1:consequent}$\Rightarrow$~\eqref{eq:lambda:many:consequent:tight}
  follows by setting
  $q \defeq \lceil \max(\lambda_1,\ldots,\lambda_k)\rceil \in \N$ and
  noting that $\bm c_i$ is tight, hence $\bm c_i \cdot \bm h \leq 0$
  and therefore
  $\lambda_i \bm c_i \cdot \bm h \geq q \bm c_i \cdot \bm h$.

\begin{corollary}
  \label{cor:lambda:1:tight} Consider a conditional
  inequality~\eqref{eq:cond:ii}.  If all antecedents are tight, then
  the corresponding decision problem $\aecondiip$ is in $\Pi^0_3$
\end{corollary}

\begin{proof}
  Based on our discussion, the inequality~\eqref{eq:cond:ii} is
  equivalent to condition~\eqref{eq:lambda:many:consequent:tight},
  which is of the form $\forall p \exists q \forall \bm h$.  Replace
  $\bm h$ with a representable entropic vector $\bm h^\Omega$, as in
  the proof of Theorem~\ref{th:co:re}, and it becomes
  $\forall p \exists q \forall \bm h^\Omega$, placing it in $\Pi^0_3$.
\end{proof}

Recall that the implication problem for CI is a special case of a
conditional inequality with tight antecedents.  We have seen in
Theorem~\ref{th:ci:in:pi1} that the {\em entropic} version of the CI
implication problem is in $\Pi^0_1$;
Corollary~\ref{cor:lambda:1:tight} proves that the {\em almost
  entropic} version is in $\Pi^0_3$.

Consider any conditional inequality~\eqref{eq:cond:ii} where the
antecedents are tight.  If this inequality holds for all almost
entropic functions, then it can be proven by proving a family of
(unconditional) inequalities~\eqref{eq:lambda:many:consequent:tight}.
In fact, some conditional inequalities in the literature have been
proven precisely in this way.  For example, consider the CI
implication \eqref{eq:k:s} (Sec.~\ref{subsubsection:cond:iip}), and
replace each antecedent $I_h(\bm Y; \bm Z|\bm X)=0$ with
$-I_h(\bm Y; \bm Z|\bm X) \geq 0$.  By
Eq.~\eqref{eq:lambda:many:consequent:tight}, the following condition
holds: $\forall p \in \N, \exists q \in \N$ such that

{\footnotesize
\begin{align}
q(I_h(C;D \suchthat A)+I_h(C;D \suchthat B)+I_h(A;B)+I_h(B;C \suchthat
  D)) + \frac{1}{p}  h(ABCD)
& \geq I_h(C;D) \label{eq:temp5}
\end{align}
}

Thus, in order to prove \eqref{eq:k:s}, it suffices to prove
\eqref{eq:temp5}.  Mat{\'{u}}{\v s}'s inequality~\eqref{eq:matus}
provides precisely the proof of \eqref{eq:temp5} (by setting
$k \defeq p$,
$q \defeq \max(\left\lceil \frac{k+3}{2} \right\rceil, 1)$, and
observing that $I_h(B;D \suchthat C) \leq h(ABCD)$).

{\bf Antecedents Have Slack} Next, we consider the case when the
antecedents have slack, which is a recursively enumerable condition.  In that case,
condition~\eqref{eq:cor:lambda:1:consequent} is equivalent to:
\begin{align}
  \exists \lambda_1 \geq 0, \cdots, \exists \lambda_k \geq 0, \forall \bm h \in \aentropic: &&  \bm c \cdot \bm h \geq \sum_{i \in [k]} \lambda_i \bm c_i \cdot \bm h\label{eq:lambda:many:consequent:slack}
\end{align}
In other words, we have proven the following result of independent
interest: any conditional implication with slack is essentially
unconditioned.  However, we cannot immediately use
\eqref{eq:lambda:many:consequent:slack} to prove complexity bounds for
$\aebic(F)$, because the $\lambda_i$'s in
\eqref{eq:lambda:many:consequent:slack} are not necessarily rational
numbers.  When we derived Eq.~\eqref{eq:lambda:many:consequent:tight}
we used the fact that the antecedents are tight, hence
$\bm c_i \cdot \bm h \leq 0$, hence we could replace the $\lambda_i$'s
with some natural number $q$ larger than all of them.  But now, the
sign of $\bm c_i \cdot \bm h$ is unknown.  We prove below that, under
a restriction called {\em group balance}, the $\lambda_i$'s can be
chosen in $\Q$, placing the decision problem in $\Sigma^0_2$.  Group
balance generalizes Chan's notion of a {\em balanced inequality},
which we review below.  In \conferenceorfull{the full version~\cite{arxiv-version}}{Appendix~\ref{sec:th:rationals}} we give
evidence that some restriction is necessary to ensure the
$\lambda_i$'s are rationals\conferenceorfull{}{ (Example~\ref{ex:s2:irrational})}, and also
show that every conditional inequality can be strengthened to be group
balanced\conferenceorfull{}{ (Prop~\ref{prop:strongly:balanced})}.

A vector $\bm h \in \R^{2^n}$ is called {\em modular} if
$h(\bm X)+h(\bm Y) = h(\bm X \cup \bm Y)+h(\bm X \cap \bm Y)$ forall
sets of variables $\bm X, \bm Y \subseteq \bm V$.
Every non-negative modular function is entropic~\cite{yeung2012book},
and is a non-negative linear combination of the {\em basic modular
  functions} $\bm h^{(1)}, \ldots, \bm h^{(n)}$, where
$h^{(j)}(\alpha) \defeq 1$ when $j \in \alpha$ and is
$h^{(j)}(\alpha) \defeq 0$ otherwise.
%
Chan~\cite{terence-chan-balanced-ii-2004} called an inequality
$\bm c \cdot \bm h \geq 0$ {\em balanced} if
$\bm c \cdot \bm h^{(j)} = 0$ for every $j \in [n]$.  He proved that
{\em any} valid inequality can be strengthened to a balanced one.
More precisely: $\bm c \cdot \bm h \geq 0$ is valid iff
$\bm c \cdot \bm h^{(i)} \geq 0$ for all $i \in [n]$ and
$\bm c \cdot \bm h - \sum_i (\bm c \cdot \bm h^{(i)}) h(X_i \suchthat
X_{[n]-\set{i}}) \geq 0$ is valid; notice that the latter inequality
is balanced.  For example, $h(XY) + h(XZ) - h(X) - h(XYZ) \geq 0$ is
balanced, while $h(XY) - h(X) \geq 0$ is not balanced, and can be
strengthened to $h(XY)-h(X)-h(Y|X) \geq 0$.  We generalize Chan's
definition:

\begin{definition} \label{def:group:balanced} Call a set
  $\set{\bm d_1, \ldots, \bm d_k} \subseteq \R^{2^n}$ {\em group
    balanced} if (a) $\rank(\bm A)=k-1$ where $\bm A$ is the
  $k \times n$ matrix $A_{ij} = \bm d_i \cdot \bm h^{(j)}$, and (b)
  there exists a non-negative modular function $\bm h^{(*)} \neq 0$
  such that $\bm d_i \cdot \bm h^{(*)} = 0$ for all $i$.
\end{definition}

If $k=1$ then $\set{\bm d_1}$ is group balanced iff $\bm d_1$ is
balanced, because the matrix $\bm A$ has a single row
$\left(\bm d\cdot \bm h^{(1)} \cdots \bm d \cdot \bm h^{(n)}\right)$,
and its rank is 0 iff all entries are 0.  We prove in \conferenceorfull{\cite{arxiv-version}}{Appendix~\ref{sec:th:rationals}}:

\begin{theorem} \label{th:rationals} Consider a group balanced set of $n$
  vectors with rational coefficients,
  $D = \set{\bm d_1, \ldots, \bm d_n} \subseteq \Q^{2^n}$.  Suppose
  the following condition holds:
  \begin{align}
\exists \lambda_1 \geq 0, \cdots, \exists \lambda_n \geq 0, \sum_{i \in [n]} \lambda_i=1, \forall \bm h \in \aentropic: && \sum_{i\in[n]} \lambda_i \bm d_i \cdot \bm h \geq 0 \label{eq:lambda:local}
  \end{align}
  Then there exists rational $\lambda_1, \ldots, \lambda_k \geq 0$
  with this property.
\end{theorem}

This implies that, if $\bm c_1, \ldots, \bm c_k$ have slack and
$\set{\bm c, -\bm c_1, \ldots, -\bm c_k}$ is group balanced, then
there exist rational $\lambda_i$'s for
inequality~\eqref{eq:lambda:many:consequent:slack}.  In particular:

\begin{corollary} \label{cor:lambda:1:slack} Consider a conditional
  inequality~\eqref{eq:cond:ii}.  If the antecedents have slack and
  $\set{\bm c, -\bm c_1, \ldots, -\bm c_k}$ is group balanced, then
  the corresponding decision problem is in $\Sigma^0_2$.
\end{corollary}

We end this section by illustrating with an example:

\begin{example}
  \label{ex:conditional:with:slack} Consider the following conditional
  inequality:

  {\footnotesize
  \begin{align}
    h(XYZ) + h(X) \geq 2h(XY) \wedge h(XYZ) + h(Y) \geq 2h(YZ)\ \ \Rightarrow\ \ 2h(XZ)\geq h(XYZ)+h(Z)
    \label{eq:kopparty:rossman:cond}
  \end{align}
}

The antecedents have slack, because, by setting\footnote{Where
  $\bm h^{(X)}$ denotes the basic modular function at $X$,
  i.e. $h^{(X)}(X)=1$, $h^{(X)}(Y)=h^{(X)}(Z)=0$.}
$\bm h \defeq 2\bm h^{(X)} + \bm h^{(Z)}$, both antecedents become
strict inequalities: $h(XYZ)+h(X)-2h(XY) = 3+2 - 4 > 0$ and
$h(XYZ)+h(Y)-2h(YZ)=3+0-2 > 0$.  To check validity, we prove in
Example~\ref{ex:kopparty:rossman:sum} the following inequality:
  \begin{align*}
    (2h(XY)-h(XYZ)-h(X))+(2h(YZ)-h(XYZ)-h(Y))+(2h(XZ)-h(XYZ)-h(Z))\geq & 0
  \end{align*}
  and this immediately implies \eqref{eq:kopparty:rossman:cond}.

  Consider now the following set
  $D = \set{\bm d_1, \bm d_2, \bm d_3}$, where the vectors
  $\bm d_1, \bm d_2, \bm d_3$ represent the expressions
  $2h(XY)-h(XYZ)-h(X)$, $2h(YZ)-h(XYZ)-h(Y)$, and $2h(XZ)-h(XYZ)-h(Z)$
  respectively.  We prove that $D$ is group balanced.  To check
  condition (a) of Def.~\ref{def:group:balanced} we verify that the
  matrix $\bm A$ has rank 2; in our example the matrix is
  $\bm A = \left(
    \begin{array}{ccc}
      0&1&-1 \\ -1&0&1 \\ -1&1&0
    \end{array}
  \right)$ and its rank is 2 as required.  To check condition (b), we
  define $\bm h^{(*)} = h^{(X)}+h^{(Y)}+h^{(Z)}$ and verify that
  $\bm d_1 \cdot \bm h^{(*)} = \bm d_2 \cdot \bm h^{(*)} = \bm d_3
  \cdot \bm h^{(*)} = 4-3-1 = 0$.  Thus, $D$ is group balanced.
\end{example}

\subsection{Discussion on the Decidability of  $\maxiip$}

A proof of the decidability of $\maxiip$ would immediately imply that
the domination problem $\bm A \preceq \bm B$ for acyclic structures
$\bm B$ is also decidable~\cite{DBLP:journals/corr/abs-1906-09727}.
It is currently open whether $\maxiip$ is decidable, or even if the
special case $\iip$ is decidable.  But what can we say about the
domination problem if $\iip$ were decidable?  Theorem~\ref{th:co:re}
only says that both problems are in $\Pi^0_1$, and does not tell us
anything about $\maxiip$ if $\iip$ were decidable.  We prove here
that, the decidability of $\iip$ implies the decidability of
group-balanced $\maxiip$.  We start with a result of general interest,
which holds even for conditional Max-Information constraints.

\begin{theorem} \label{th:lambdas:real} The following two statements are
  equivalent:
  \begin{align}
      \forall \bm h \in \aentropic: && \bigwedge_{i\in [k]} \bm c_i \cdot \bm h \geq 0 \Rightarrow
    \bigvee_{j\in [m]} \bm d_j \cdot \bm h \geq 0\label{eq:lambda:1}\\
  \exists \lambda_1, \ldots, \lambda_m \geq 0, \sum_j \lambda_j = 1, \forall \bm h \in
    \aentropic: && \bigwedge_{i\in [k]} \bm c_i \cdot \bm h \geq 0 \Rightarrow \sum_{j\in [m]} \lambda_j\bm d_j \cdot \bm h \geq 0\label{eq:lambda:2}
  \end{align}
\end{theorem}

The theorem says that every $\max$-inequality is essentially a linear
inequality.  The proof of ~\eqref{eq:lambda:2} $\Rightarrow$
~\eqref{eq:lambda:1} is immediate; we prove the reverse in
\conferenceorfull{\cite{arxiv-version}}{Appendix~\ref{app:max:inequality}}. As before, we don't know whether
these coefficients $\lambda_i$ can be chosen to be rational numbers in
general, but by Theorem~\ref{th:rationals} this is the case when
$\set{\bm c_1, \ldots, \bm c_k}$ is group-balanced, and this implies:

\begin{corollary} \label{cor:maxiip:iip}
  The $\maxiip$ problem where the inequalities
  $\bm c_1, \ldots, \bm c_n$ are group balanced is Turing equivalent
  to the $\iip$ problem.
\end{corollary}

\begin{proof}
  We describe a Turing reduction from $\maxiip$ to $\iip$.  Consider a
  $\maxiip$ problem,
  $\bigvee_{j\in [m]} (\bm c_j \cdot \bm h \geq 0)$.  We run two
  computations in parallel.  The first computation iterates over all
  representable spaces $\Omega$, and checks whether
  $\bigwedge_j (\bm c_j \cdot \bm h^\Omega < 0)$; if we find such a
  space then we stop and we return {\em false}.  If the inequality is
  invalid then this computation will eventually terminate because in
  that case there exists a representable counterexample $\Omega$.  The
  second computation iterates over all $m$-tuples of natural numbers
  $(\lambda_1, \ldots, \lambda_m) \in \N^m$ and checks
  $\forall \bm h \in \entropic, \sum_j \lambda_j \bm c_j \cdot \bm h
  \geq 0$ by using the oracle for $\iip$: if it finds such
  $\lambda_j$'s, then it stops and returns {\em true}.  If the
  inequality is valid then this computation will eventually terminate,
  by Theorems~\ref{th:lambdas:real} and~\ref{th:rationals}.
\end{proof}

We illustrate with an example.

\begin{example} \label{ex:kopparty:rossman:sum} Consider Kopparty and
  Rossman's inequality \eqref{eq:kopparty:rossman}, which can be
  stated as $\max(\bm c_1, \bm c_2, \bm c_3) \geq 0$, where
  $\bm c_1, \bm c_2, \bm c_3$ define the three expressions in
  \eqref{eq:kopparty:rossman}.  To prove that it is valid, it suffices
  to prove that their sum is $\geq 0$; we show this briefly
  here\footnote{We apply submodularity: $h(XY)-h(X) \geq h(XYZ)-h(XZ)$
    etc.}:
  \begin{align*}
    & (2h(XY) - h(X))  +  (2h(YZ) - h(Y)) +  (2h(XZ) - h(Z)) - 3h(XYZ) \\
    &= (h(XY)+h(YZ) + h(XZ))
    + (h(XY) - h(X))  +  (h(YZ) - h(Y)) +  (h(XZ) - h(Z)) \\&\quad- 3h(XYZ) \\
    &\geq (h(XY)+h(YZ) + h(XZ)) +
    (h(XYZ)-h(XZ)) + (h(XYZ)-h(XY)) \\
    & \qquad +(h(XYZ)-h(YZ))-3h(XYZ)
    =0
  \end{align*}
  Theorem~\ref{th:lambdas:real} proves that {\em any} max-inequality
  necessarily follows from such a linear inequality; we just have to
  find the right $\lambda_i$'s.  In this example, the set
  $\bm c_1, \bm c_2, \bm c_3$ is group balanced (as we showed in
  Example~\ref{ex:conditional:with:slack}), therefore there exists
  rational $\lambda_i$'s; indeed, our choice here is
  $\lambda_1=\lambda_2=\lambda_3=1$.
\end{example}

\section{The Recognizability Problems}

\label{sec:recognizability}

We study here two problems that are the dual of the Boolean
information constraint problem.  The {\em entropic-recognizability
  problem} takes as input a vector $\bm h$ and checks if
$\bm h \in \entropic$.  The {\em almost-entropic-recognizability
  problem} checks if $\bm h \in \aentropic$.  We will prove that the
latter is in $\Pi^0_2$, and leave open the complexity of the former.

Before we define these problems formally, we must first address the
question of how to represent the input $\bm h$.  One possibility is to
represent $\bm h$ as a vector of rational numbers, but this is
unsatisfactory, because usually entropies are not rational numbers.
Instead, we will allow a more general representation.  To justify it,
assume first that $\bm h$ were given by some representable space
$\Omega$ (Sec.~\ref{sec:unconditional:bic}), where all probabilities
are rational numbers.  In that case, every term $p_i \log p_i$ in the
definition of the entropy can be written as $\log (p_i^{p_i})$, hence
the quantity $h(\bm X)$ has the form
$h(\bm X) = \log \prod_i p_i^{p_i}$.  In general, any product
$\prod_i m_i^{n_i}$ where $m_i, n_i \in Q$, for $i=1,n$, can be
rewritten as $\left(\frac{a}{b}\right)^{\frac{1}{c}}$, where
$a, b, c \in \N$.  Indeed, writing $m_i = u_i/v_i$ and $n_i = s_i/t_i$
where $u_i, v_i, s_i, t_i \in \N$, we have:
\begin{align*}
\prod_i \left(\frac{u_i}{v_i}\right)^{\frac{s_i}{t_i}} = &
\prod_i \left(\frac{u_i^{s_i}}{v_i^{s_i}}\right)^{\frac{1}{t_i}} =
\left(\prod_i \frac{u_i^{s_i\cdot \prod_{j\neq i}t_j}}{v_i^{s_i\cdot \prod_{j\neq i}t_j}}\right)^{\frac{1}{\prod_i t_i}} =
\left(\frac{a}{b}\right)^{\frac{1}{c}} && a, b, c \in \N
\end{align*}
Justified by this observation, we assume that the input to our problem
consists of three vectors $(a_{\bm X})_{\bm X \subseteq \bm V}$,
$(b_{\bm X})_{\bm X \subseteq \bm V}$, and
$(c_{\bm X})_{\bm X \subseteq \bm V}$ in $\N^{2^n}$, with the
convention that
$h(\bm X) \defeq \frac{1}{c_{\bm X}} \log \frac{a_{\bm X}}{b_{\bm
    X}}$.  Thus, we do not assume that these vectors come from a
representable space $\Omega$, we only assume their entropies can be
represented in this form.

\begin{definition}[(Almost-)Entropic Recognizability
  Problem] \label{def:recognizability} Given natural
  numbers $(a_{\bm X})_{\bm X \subseteq \bm V}$,
  $(b_{\bm X})_{\bm X \subseteq \bm V}$ and
  $(c_{\bm X})_{\bm X \subseteq \bm V}$, check whether the vector
  $h({\bm X}) \defeq \frac{1}{c_{\bm X}} \log \frac{a_{\bm X}}{b_{\bm
      X}}$, $\bm X \subseteq \bm V$, represents an entropic vector, or
  an almost-entropic vector.
\end{definition}

Our result in this section is (see \conferenceorfull{\cite{arxiv-version}}{Appendix~\ref{app:almost:recognizability:pi2}} for a proof):

\begin{theorem} \label{th:almost:recognizability:pi2} The almost entropic
  recognizability problem is in $\Pi^0_2$.
\end{theorem}

We end with a brief comment on the complexity of the
entropic-recognizability problem: given $\bm h$ (represented as in
Def.~\ref{def:recognizability}) check if $\bm h \in \entropic$.
Consider the following restricted form of the problem: check if
$\bm h$ is the entropic vector of a representable space $\Omega$
(i.e. finite space with rational probabilities).  This problem is in
$\Sigma^0_1$, because one can iterate over all representable spaces
$\Omega$ and check that their entropies are those required.  However, in
the general setting we ask whether {\em any} finite probability space
has these entropies, not necessarily one with rational probabilities.
This problem would remain in $\Sigma^0_1$ if the theory of reals with
exponentiation were decidable.  Recall that Tarski's theorem states
that the theory of reals $\mathrm{FO}(\R, 0, 1, +, *)$ is decidable.  A major
open problem in model theory is whether the theory remains decidable
if we add exponentiation.  If that were decidable, then the
entropic-recognizability problem would be in $\Sigma^0_1$.  To see this,
consider the following semi-decision problem.  Iterate over
$N=1, 2, 3, \ldots$ and for each $N$ check if there exists a
probability space whose active domain has size $N$ (thus, there are
$N^n$ outcomes, where $n =|\bm V|$ is the number of variables) and whose
entropies are precisely those given.  This statement that can be
expressed using the exponential function (which we need in order to
express the entropy as $\sum_i p_i \log p_i$).  If there exists any
finite probability space with the required entropies, then this
procedure will find it; otherwise it will run forever, placing the
problem in $\Sigma^0_1$.

\section{Discussion}

\label{sec:discussion}

\textbf{CI Implication Problem} The implication problem for
Conditional Independence statements has been extensively studied in
the literature, but its complexity remains an open problem.  It is not
even known whether this problem is
decidable~\cite{GeigerPearl1993,DBLP:journals/ijar/NiepertGG10,DBLP:journals/ai/NiepertGSG13}.
Our Theorem~\ref{th:ci:in:pi1} appears to be the first upper bound on
the complexity of the CI implication problem, placing it in $\Pi^0_1$.
Hannula et al.~\cite{DBLP:journals/corr/abs-1812-05873} prove that, if
all random variables are restricted to be binary random variables,
then the CI implication problem is in EXPSPACE; the implication
problem for binary random variables differs from that for general
discrete random variables; see the discussion
in~\cite{GeigerPearl1993}.

\textbf{Finite, infinite, continuous random variables.} In this paper,
all random variables have a finite domain.  There are two alternative
choices: discrete random variables (possibly infinite), and continuous
random variables.  The literature on entropic functions has mostly
alternated between defining entropic functions over finite random variables, or
over discrete infinite random variables with finite entropy.  For example discrete
(possibly infinite) random variables are considered by Zhang and
Yeung, ~\cite{DBLP:journals/tit/ZhangY97}, by Chan and
Yeung~\cite{DBLP:journals/tit/ChanY02}, and by
Chan~\cite{terence-chan-balanced-ii-2004}, while random variables with
finite domains are considered by Mat{\'{u}}{\v
  s}~\cite{Matus94probabilisticconditional,DBLP:conf/isit/Matus07} and
by Kaced and Romashchenko~\cite{DBLP:journals/tit/KacedR13}.  The
reason for this inconsistency is that for information inequalities the
distinction doesn't matter: every entropy of a set of discrete random
variables can be approximated arbitrarily well by the entropy of a set
of random variables with finite domain, and
Prop.~\ref{prop:rational:p} extends immediately to discrete random
variables\footnote{The idea of the proof relies on the fact that every
  entropy is required to converge, i.e.
  $h(X_\alpha) = \sum_i p_i \log 1/p_i$, hence there exists a finite
  subspace of outcomes $\set{1,2,\ldots, N}$ for which the sum is
  $\varepsilon$-close to $h(X_\alpha)$.  The union of these spaces
  over all $\alpha \subseteq [n]$ suffices to approximate $h$ well
  enough.}.  However, the distinction is significant for conditional
inequalities, and here the choice in the literature is always for
finite domains.  For example, the implication problem for conditional
independence, i.e. the graphoid axioms, is stated for finite
probability spaces by Geiger and Pearl~\cite{GeigerPearl1993}, while
Kaced and Romashchenko~\cite{DBLP:journals/tit/KacedR13} also use
finite distributions to prove the existence of conditional
inequalities that hold over entropic but fail for almost-entropic
functions.  One could also consider continuous distributions, whose
entropy is $\int p(x) \log(1/p(x))dx$, where $p$ is the probability
density function.  Chan~\cite{terence-chan-balanced-ii-2004} showed
that an information inequality holds for all continuous distributions
iff it is balanced and it holds for all discrete distributions.  For
example, $h(X) \geq 0$ is not balanced, hence it fails in the
continuous, because the entropy of the uniform distribution in the
interval $[0,c]$ is $\log c$, which is $<0$ when $c < 1$.

\textbf{Strict vs.\ non-strict inequalities.} The literature on
information inequalities always defines inequalities using $\geq 0$,
in which case validity for entropic functions is the same as validity
for almost entropic functions.  One may wonder what happens if one
examines strict inequalities $\bm c \cdot \bm h > 0$ instead.
Obviously, each such inequality fails on the zero-entropic vector, but
we can consider the conditional version
$\bm h \neq 0 \Rightarrow \bm c \cdot \bm h > 0$, which we can write
formally as $\bm c \cdot \bm h \leq 0 \Rightarrow h(\bm V) \leq 0$.
This a special case of a conditional inequality as discussed in this
paper.  An interesting question is whether for this special case
$\entropic$-validity and $\aentropic$-validity coincide; a negative
answer would represent a significant extension of Kaced and
Romashchenko's result~\cite{DBLP:journals/tit/KacedR13}.



\bibliography{main}

\conferenceorfull{}{
    \appendix
    \section{Background on Cones}

\label{appendix:cones}


We will employ basic facts about closed, convex
cones~\cite{DBLP:books/cu/BV2014,DBLP:journals/kybernetika/Studeny93},
which we review briefly in this section.

Fix a set $S \subseteq \R^m$.  We denote by $\bar S$ its topological
closure.  The set $S$ is {\em convex} if it is closed under taking
convex combination, i.e. $\bm x, \bm y \in S$ and $\theta \in [0,1]$
imply $\theta \bm x + (1-\theta)\bm y \in S$.  The set $S$ is a
(Euclidean) {\em cone} if it is closed under taking non-negative
multiple, i.e. $\bm x \in S$ and $\theta \geq 0$ imply
$\theta \bm x \in S$.  We use $\con(S)$ to denote its {\em conic
  hull}, i.e. the set of conic combinations $\sum_i \theta_i \bm x_i$
where $\theta_i \geq 0$ and $\bm x_i \in S$.  It is easy to see that
conic hulls are convex.

For any set $K \subseteq \R^m$, its {\em dual cone} $K^*$ is defined
by
\begin{align}
    K^* &:= \setof{\bm y}{\bm x\cdot \bm y \geq 0, \forall \bm x \in K}
\end{align}
If $\bm x \in \R^m$ then we denote $\bm x^* = \setof{\bm y}{\bm x \cdot \bm y \geq 0}$ for
short.

It is not hard to see that $K^*$ is always a closed, convex cone
(regardless of whether $K$ is closed or convex or even a cone).  For
any two sets $K, L\subseteq \R^m$ it holds that $K \subseteq L^*$ iff
$K^* \supseteq L$ (i.e. $(-)^*$ forms an antitone Galois connection).
It is also known that, taking duality twice, $K^{**}$ is the closure
of the smallest convex cone containing $K$, i.e.
$K^{**} =\cl{\con(K)}$; in particular, $K^{**}=K$ iff $K$ is a closed
convex cone.

A cone $K$ is called {\em pointed} if $\bm x \in K$ and
$-\bm x \in K$ imply $\bm x = \bm 0$; in other words, $K$ is pointed if it contains no line.
If $K$ has a non-empty interior, then its dual $K^*$ is pointed. If $K$ is a cone and its
closure is pointed, then $K^*$ has a non-empty interior.
For any
sets $K, L \subseteq \R^m$, it is easy to see that $(K \cup L)^* = K^* \cap L^*$.

\section{Proof of Theorem~\ref{th:lambda:1}}

\label{app:proof:th:lambda:eps}

We need the following non-obvious lemma due to Studen{\'{y}}:

\begin{lemma} \label{lemma:studeny} \cite[Lemma
  1]{DBLP:journals/kybernetika/Studeny93} Let $L \subseteq \R^m$ be a
  closed, convex cone, and $\bm y \in \R^m$ be a vector s.t.
  $\bm y \not\in -L$.  Then $\con(L \cup \set{\bm y})$ is closed.
\end{lemma}

\begin{example} \label{ex:s2} The condition $\bm y \not\in -L$ is
  necessary, as illustrated by the following example, also
  from~\cite{DBLP:journals/kybernetika/Studeny93}.  Let
  \begin{align}
    L = & \setof{(a,b,c)}{a, b, c \in \R, a \geq 0, c \geq 0, ac \geq b^2} \label{eq:s2}
  \end{align}
  One can check that this is a closed, convex cone\footnote{$L$ is
    isomorphic to the cone $S^2$ of positive semi-definite symmetric
    $2 \times 2$ matrices.}.  Notice that $c=0$ implies $b=0$.  Let
  $\bm y = (-1,0,0)$, then
  $\con(L \cup \set{\bm y})=\setof{(a,b,c)}{a,b,c\in \R, c \geq 0
    \wedge (c=0\Rightarrow b=0)}$.  Then the sequence of vectors
  $\bm z_n = (0,1,1/n)$ is in $\con(L \cup \set{\bm y})$, but their
  limit $(0,1,0)$ is not, proving that $\con(L \cup \set{\bm y})$ is
  not closed.
\end{example}
%
%

Fix a cone $K \subseteq \R^m$.  We say that a set of vectors
$\set{\bm y_1, \ldots, \bm y_k} \subseteq \R^m$ has {\em slack}
w.r.t. $K$ if there exists $\bm x \in K$ such that
$\bm y_i \cdot \bm x > 0$ for $i=1,k$.  We prove the following lemma,
which is an extension of a result
in~\cite{DBLP:journals/corr/abs-1812-09987}:

\begin{lemma} \label{lemma:K:lambda:1}
  Let $K\subseteq \R^m$ be a closed convex
  cone.  Then the following statements are equivalent:
  \begin{align}
      \forall \bm x \in K: & \bigwedge_{i\in [k]} (\bm x \cdot \bm y_i \geq 0) \Rightarrow   (\bm x \cdot \bm y \geq 0)\label{eq:K:lambda:1}\\
      \forall \varepsilon > 0, \exists \lambda \geq 0, \forall \bm x \in K: &
\bigwedge_{i\in [k-1]} (\bm x \cdot \bm y_i \geq 0) \Rightarrow
(\bm x \cdot \bm y + \varepsilon ||\bm x||_2 \geq \lambda \bm x\cdot \bm y_k)\label{eq:K:lambda:2}
  \end{align}
  Moreover, if $\set{\bm y_1, \ldots, \bm y_k}$ has slack w.r.t. $K$
  then we can set $\varepsilon=0$ in Eq.\eqref{eq:K:lambda:2}.
\end{lemma}

\begin{proof}
  The implication~\eqref{eq:K:lambda:2} $\Rightarrow$
  \eqref{eq:K:lambda:1} is immediate, hence we
  prove~\eqref{eq:K:lambda:1} $\Rightarrow$ \eqref{eq:K:lambda:2}.
  For every $i$, the statement $\bm x \cdot \bm y_i \geq 0$ is
  equivalent to $\bm x \in \bm y_i^*$.  Denote
  $L \defeq K \cap \bigcap_{i \in [k-1]} y_i^*$ and notice that this
  is a closed, convex cone.  We start by proving that
  Condition~\eqref{eq:K:lambda:1} is equivalent to
  $\bm y \in \cl{\con(L^* \cup \set{\bm y_k})}$:
  \begin{align*}
    \mbox{Condition~\eqref{eq:K:lambda:1}} & \Leftrightarrow L \cap \bm y_k^* \subseteq \bm y^* \\
& \Leftrightarrow L^{**} \cap \bm y_k^* \subseteq \bm y^* & \mbox{$L$ is closed, convex, hence $L^{**}=L$}\\
& \Leftrightarrow (L^* \cup \set{\bm y_k})^* \subseteq \bm y^* & \mbox{$A^* \cap B^* = (A \cup B)^*$} \\
& \Leftrightarrow \bm y \in (L^* \cup \set{\bm y_k})^{**} & \mbox{$A \subseteq B^*$ iff $A^* \supseteq B$}\\
& \Leftrightarrow \bm y \in \cl{\con(L^* \cup \bm y_k)} & \mbox{$A^{**} = \cl{\con(A)}$}
  \end{align*}
  Consider first the case when the set
  $\set{\bm y_1, \ldots, \bm y_k}$ has slack in $K$, and let
  $\bm x_0 \in K$ be such that $\bm x_0 \cdot \bm y_i >0$ for all
  $i=1,k$; in particular, $\bm x_0 \in L$ and
  $\bm x_0 \cdot \bm y_k > 0$.  This implies that
  $\bm y_k \not\in -L^*$, hence, by Lemma~\ref{lemma:studeny},
  $\con(L^* \cup \set{\bm y_k})$ is closed.  Therefore we have that
  $\bm y \in \con(L^* \cup \set{\bm y_k})$, hence there exist
  $\bm z \in L^*$ and $\lambda \geq 0$ such that
  $\bm y = \bm z + \lambda \bm y_k$.  To prove
  condition~\eqref{eq:K:lambda:2}, it suffices to show that
  $\forall \bm x \in L$,
  $\bm x \cdot \bm y \geq \lambda \bm x \cdot \bm y_k$. This follows
  from the fact that $\bm x \cdot \bm z \geq 0$, which implies that
  $\bm x \cdot \bm y = \bm x \cdot \bm z + \lambda \bm x \cdot \bm y_k
  \geq \lambda \bm x \cdot \bm y_k$ as required.

  Consider now the general case, when
  $\bm y \in \cl{\con(L^* \cup \set{\bm y_k})}$.  Then, for every
  $\varepsilon > 0$ there exists
  $\bm y' \in \con(L^* \cup \set{\bm y_k})$ such that, denoting
  $\bm \delta \defeq \bm y' - \bm y$, we have
  $||\bm \delta||_2 < \varepsilon$.  Applying the argument above to
  $\bm y'$ instead of $\bm y$, we obtain
  $\bm x \cdot \bm y' \geq \lambda \bm x \cdot \bm y_k$.  On the other
  hand,
  $\bm x \cdot \bm y' = \bm x \cdot \bm y + \bm x \cdot \bm \delta
  \leq \bm x \cdot \bm y + ||\bm x||_2 \cdot ||\bm \delta||_2 \leq \bm
  x \cdot \bm y + \varepsilon||\bm x||_2$.
\end{proof}

\begin{example}
  \label{ex:studeny} We show that the error term
  $\varepsilon ||\bm x||_2$ in Condition~\eqref{eq:K:lambda:2} is
  necessary in general.  For that, consider the cone $L$ in
  Eq.~\eqref{eq:s2}.  It satisfies the condition
  $\forall (a,b,c) \in L: a \leq 0 \Rightarrow b \leq 0$.  Indeed,
  $a \leq 0$ is equivalent to $a=0$, thus $b^2 \leq ac = 0$ implying
  $b=0$, in particular $b \leq 0$.  By writing the implication as
  $-a \geq 0 \Rightarrow -b \geq 0$, Lemma~\ref{lemma:K:lambda:1} says
  that $\forall \varepsilon > 0, \exists \lambda \geq 0$ such that
  $\forall (a,b,c) \in L: -b + \varepsilon ||(a,b,c)||_2 \geq \lambda
  (-a)$ or, equivalently
  $b \leq \varepsilon ||(a,b,c)||_2 + \lambda a$.  When
  $\varepsilon=0$ then this condition fails for any choice of
  $\lambda\geq 0$, for example it fails on the vector
  $(1, 1+\lambda, (1+\lambda)^2) \in L$.  On the other hand, if
  $\varepsilon > 0$, then the condition holds, for example we can
  choose $\lambda = 1/\varepsilon$ and we obtain
  $\varepsilon ||(a,b,c)||_2 + \lambda a \geq \varepsilon c + \lambda
  a \geq 2 \sqrt{ac} \geq 2|b| \geq b$.
\end{example}

We can now prove Theorem~\ref{th:lambda:1}.

\begin{proof} (of Theorem~\ref{th:lambda:1}) It remains to prove that
  \eqref{eq:lambda:1:antecedent} implies
  \eqref{eq:lambda:1:consequent} This follows from
  Lemma~\ref{lemma:K:lambda:1} applied to $K \defeq \aentropic$, which
  is a closed, convex cone, see~\cite{yeung2012book}.  By renaming the
  vectors $\bm x, \bm y_i, \bm y$ in Lemma~\ref{lemma:K:lambda:1} to
  $\bm h, \bm c_i, \bm c$ in Theorem~\ref{th:lambda:1}, and using
  $\varepsilon/2^n$ instead of $\varepsilon$, we use the Lemma to
  argue that \eqref{eq:lambda:1:antecedent} implies:
  \begin{align*}
        \forall \varepsilon>0,\exists \lambda>0, \forall \bm h \in \aentropic: &
    \bigwedge_{i \in [k-1]} \bm c_i \cdot \bm h \geq 0
    \Rightarrow \bm c \cdot \bm h+ \frac{1}{2^n} \varepsilon ||\bm h||_2 \geq \lambda \bm c_k \cdot \bm h
  \end{align*}
  Condition \eqref{eq:lambda:1:consequent} follows from the fact that
  $h([n]) \geq \frac{1}{2^n} \sum_{\alpha \subseteq [n]} h(\alpha) = \frac{1}{2^n} ||\bm h||_1 \geq \frac{1}{2^n} ||\bm
  h||_2$.
\end{proof}

We end this section by stating an obvious consequence of
Lemma~\ref{lemma:K:lambda:1}: by applying it repeatedly, we can move
all antecedents to the consequent, and obtain:

\begin{corollary} \label{cor:K:lambda:1}
  Let $K\subseteq \R^m$ be a closed convex
  cone, and assume that $\set{\bm y_1, \ldots, \bm y_k}$ has slack
  w.r.t. $K$.  Then the following are equivalent:
  \begin{align}
      \forall \bm x \in K: & \bigwedge_{i\in [k]} (\bm x \cdot \bm y_i \geq 0) \Rightarrow   (\bm x \cdot \bm y \geq 0)\label{eq:K:lambda:many:1}\\
      \exists \lambda_1, \ldots, \lambda_k \geq 0, \forall \bm x \in K: &
(\bm x \cdot \bm y \geq \sum_{i \in [k]} \lambda_i \bm x\cdot \bm y_i)\label{eq:K:lambda:many:2}
  \end{align}
\end{corollary}

\section{Proof of Theorem~\ref{th:rationals}}

\label{sec:th:rationals}

We prove Theorem~\ref{th:rationals} by generalizing it to arbitrary
cones.  First, we give the obvious generalization of the notion of
group balanced.  Fix a set of vectors
$\bm x_1, \ldots, \bm x_n \in \Q^m$ with rational coordinates.

\begin{definition} \label{def:group:balanced:cone} A set
    $D = \set{\bm y_1, \ldots, \bm y_k} \subseteq \R^m$ is called {\em
        group balanced} if (a) $\rank(\bm A) = k-1$ where $\bm A$ is the
    matrix $A_{ij} = \bm x_j \cdot \bm y_i$, and (b) there exists
    $\bm x^{(*)} \in \convex(\bm x_1, \ldots, \bm x_n)$ such that
    $\bm x^{(*)} \cdot \bm y_i = 0$ for all $i$.
\end{definition}

\begin{theorem} \label{th:rationals:cone} Let $K \subseteq \R^m$ be a cone
    and $D = \set{\bm y_1, \ldots, \bm y_n} \subseteq \Q^m$ be a group-balanced 
    set of rational vectors.  Suppose that the following
    condition holds:
    \begin{align}
        \exists \lambda_1 \geq 0, \ldots, \lambda_k \geq 0, \sum_i \lambda_i=1, \forall \bm x
        \in K: && \sum_i \lambda_i \bm x \cdot \bm y_i \geq 0 \label{eq:lambda:local:cone}
    \end{align}
    Then there exist rational $\lambda_i$'s with this property.
\end{theorem}

\begin{proof}
    Denote by $\Lambda \defeq \setof{(\lambda_1, \ldots, \lambda_n)}{\lambda_1\geq 0, \cdots, \lambda_n \geq 0, \forall \bm x \in K, \sum_i \lambda_i \bm  \cdot \bm y_i \geq 0}$.
    Then $\Lambda$ is a convex cone, and is $\neq \set{\bm 0}$ by
    condition \eqref{eq:lambda:local:cone}.  To prove the theorem, we
    will show $\Lambda = \setof{t \bm \lambda}{t \geq 0}$ for some
    rational vector $\bm \lambda \in \Q^n_+$.  Denote the matrix
    $\bm A \defeq (\bm y_i \cdot \bm x_j)_{ij}$; by assumption its rank
    is $n-1$.  Def.~\ref{def:group:balanced:cone} (b) implies that there
    exist $\mu_1 \geq 0, \ldots, \mu_n \geq 0$ such that, denoting
    $\bm x^{(*)} \defeq \sum_{j \in [n]} \mu_j \bm x_j$, it holds that
    $\bm x^{(*)} \cdot \bm y_i = 0$ for all $i$.  Since
    $\rank(A) = n-1$, we have $\mu_j > 0$ for all $j$.  Consider now any
    $(\lambda_1, \ldots, \lambda_n) \in \Lambda$, and let
    $\bm y^{(*)} = \sum_{i \in [n]} \lambda_i \bm y_i$.  We prove that,
    for every $j \in [n]$, $\bm x_j \cdot \bm y^{(*)}=0$.  To prove
    this, we note that $\bm x^{(*)} \cdot \bm y_i=0$ and
    $\bm x_j \cdot \bm y^{(*)} \geq 0$
    (condition~\eqref{eq:lambda:local:cone} applied to $\bm x_j$) imply:
    \begin{align*}
        0 = \sum_{i \in [n]} \lambda_i \bm x^{(*)} \cdot \bm y_i = \sum_{i, j \in [n]} \lambda_i \mu_j \bm x_j \cdot \bm y_i = \sum_{j  \in [n]} \mu_j (\bm x_j \cdot \bm  y^{(*)}) \geq 0
    \end{align*}
    If $\bm x_j \cdot \bm y^{(*)} > 0$ for some $j$ then we obtain
    $0 > 0$, a contradiction, hence $\bm x_j \cdot \bm y^{(*)} = 0$ for
    all $j$.  Let $\bm v_1, \cdots, \bm v_n$ be the column vectors of
    $\bm A$, and let $V \defeq \vspan(\bm v_1, \cdots, \bm v_n)$.  We
    have proven that $\Lambda \subseteq V^\perp$, where $V^\perp$
    denotes the orthogonal space\footnote{For any set
        $V \subseteq \R^k$, its orthogonal space is
        $V^\perp \defeq \setof{\bm u}{\forall \bm v \in V, \bm u \cdot \bm
            v = 0}$.}.  Since $\dim(V)=n-1$, we have that $\dim(V^\perp)=1$,
    in other words $V^\perp = \setof{t \bm \lambda}{t \in \R}$ for some
    non-zero vector $\bm \lambda$.  Moreover, $\bm \lambda$ can be
    chosen to be a rational vector, because $\bm v_1, \ldots, \bm v_n$
    have rational coordinates.  This proves the claim and the theorem.
\end{proof}

Next we give an example showing that some additional condition on the
vectors $\bm y_1, \ldots, \bm y_k$ is necessary to ensure that the
values $\lambda_i$'s can be chosen to be rational numbers.

\begin{example} \label{ex:s2:irrational} We show here an example where the
    values $\lambda_i$ cannot be chosen to be rational numbers.  We
    generalize Example~\ref{ex:s2} as follows.  Fix two numbers
    $\alpha, \gamma \in (0,1)$ such that $\alpha+\gamma = 1$ and
    $\frac{\alpha}{\gamma} \not\in \Q$.  Consider the cone:
    \begin{align}
        K \defeq \setof{(a,b,c))}{a \geq 0, c \geq 0, a^\alpha c^\gamma \geq b} \label{eq:irrational:cone}
    \end{align}
    (The reader may verify that $K$ is a closed, convex cone.) We claim
    that, forall $(a,b,c) \in K$, either $a \geq b$ or $c \geq b$, in
    other words $K$ satisfies the max-inequality:
    \begin{align*}
        \forall (a,b,c) \in K: && \max(a-b, c-b) \geq 0
    \end{align*}
    This follows from the inequality between the weighted arithmetic
    mean and weighted geometric mean:
    $\alpha a + \gamma c \geq a^\alpha c^\gamma \geq b = (\alpha +
    \gamma) b$, hence, if both $a < b$ and $c < b$ then
    $\alpha a + \gamma c < (\alpha + \gamma)b = b$; a contradiction.  By
    Theorem~\ref{th:lambdas:real:cone} there exist
    $\lambda_1, \lambda_2 \geq 0$, $\lambda_1 + \lambda_2 = 1$ such that
    \begin{align*}
        \forall (a,b,c) \in K: && \lambda_1(a-b) + \lambda_2(c-b) \geq 0
    \end{align*}
    Equivalently, $\forall (a,b,c) \in K$,
    $\lambda_1 a + \lambda_2 c \geq b$.  We prove that the only possible
    values are $\lambda_1 = \alpha, \lambda_2 = \gamma$, and thus no
    rational values exist.  For that we set $b = a^\alpha c^{1-\alpha}$
    in the inequality above and obtain:
    \begin{align*}
        \forall a, c \geq 0: &&& \lambda_1 a + (1-\lambda_1) c \geq a^\alpha  c^{1-\alpha} \\
        \forall a\geq 0, c>0: &&& \lambda_1 \frac{a}{c} + 1- \lambda_1 \geq  \left(\frac{a}{c}\right)^\alpha \\
        \forall x \geq 0: &&& f(x) \defeq (1 - x^\alpha) -  \lambda_1(1-x) \geq 0
    \end{align*}
    Since $f(1) = 0$ and $f(x) \geq 0$ for $x \in \R$, we must have
    $f'(1) = 0$ by Lagrange's theorem.  Thus, $-\alpha + \lambda_1 = 0$
    or $\lambda_1 = \alpha$, proving the claim.
\end{example}

Finally, we prove that the definition of strong balance is natural.
Recall Condition~\eqref{eq:lambda:local} of
Theorem~\ref{th:rationals}, which we repeat here for readability:
\begin{align}
    \exists \lambda_1 \geq 0, \ldots, \lambda_n \geq 0, \sum_{i \in [n]} \lambda_i=1, \forall \bm h \in \aentropic: & \sum_i \lambda_i \bm d_i \cdot \bm h \geq 0 \label{eq:lambda:local:repeat}
\end{align}

\begin{proposition} \label{prop:strongly:balanced} For any set of vectors
    $D = \set{\bm d_1, \ldots, \bm d_n} \subseteq \R^{2^n}$ there
    exists a strongly balanced set
    $D' = \set{\bm d_1', \ldots, \bm d_n'} \subseteq \R^{2^n}$ such
    that $D$ satisfies Condition~\eqref{eq:lambda:local:repeat} iff
    $D'$ satisfies it.
\end{proposition}

\begin{proof} (Of Prop.~\ref{prop:strongly:balanced}) Fix any set
    $D = \set{\bm d_1, \ldots, \bm d_n} \in \R^{2^n}$, and denote by
    $\bm A$ the $n \times n$ matrix
    $A_{ij} = \bm d_i \cdot \bm h^{(j)}$.  We will modify $D$ to ensure
    that the matrix $A$ satisfies Def.~\ref{def:group:balanced} First,
    we replace each $\bm d_i$ by $\bm d_i'$ obtained using Chan's
    transformation:
    \begin{align*}
        \bm d_i'\cdot \bm h = & \bm d_i \cdot \bm h - \sum_j (\bm d_i \cdot \bm h^{(j)}) h(X_j \suchthat X_{[n]-\set{j}})
    \end{align*}
    and denote by $D' \defeq \set{\bm d_1', \ldots, \bm d_n'}$.  First,
    we claim that $D$ satisfies~\eqref{eq:lambda:local:repeat} iff $D'$
    does.  This follows from Chan's theorem, since, for any values
    $\lambda_1, \ldots, \lambda_n$ that sum to 1, the expression
    $\bm d' \defeq \sum_i \lambda_i \bm d_i'$ is precisely Chan's
    transformation applied to $\bm d \defeq \sum_i \lambda_i d_i'$,
    hence $\bm d' \cdot \bm h \geq 0$ is valid iff
    $\bm d \cdot \bm h \geq 0$ is valid.  After this transformation,
    every $\bm d_i'$ is balanced, and therefore the matrix $\bm A'$
    associated to the new set $D'$ is identically 0.  Next, assume that
    $D'$ satisfies condition~\eqref{eq:lambda:local:repeat}, and let
    $\lambda_1, \ldots, \lambda_n$ be the corresponding coefficients.
    We have $\lambda_i > 0$ for all $i$ because the rank of $\bm A$ is
    $n-1$.  Replace each $\bm d_i'$ by $\bm d''_i$ where:
    \begin{align*}
        \bm d_i'' \cdot \bm h = & \bm d_i' \cdot \bm h + \frac{1}{\lambda_i}\left(n h(X_i \suchthat X_{[n]-\set{i}}) - \sum_{j\in[n]} h(X_j\suchthat X_{[n]-\set{j}})\right)
    \end{align*}
    In other words, we add the term
    $\frac{n-1}{\lambda_i}h(X_i \suchthat X_{[n]-\set{i}})$ and subtract
    all terms $\frac{1}{\lambda_i}h(X_j \suchthat X_{[n]-\set{j}})$ for
    $j\neq i$.  This transformation does not affect the
    expression~\eqref{eq:lambda:local:repeat} because
    $\sum_i \lambda_i \bm d_i' = \sum_i \lambda_i \bm d_i''$.  The new
    $n \times n$ matrix $\bm A''$ is as follows.  The diagonal is
    $A_{ii} = (n-1)/\lambda_i$, and all other entries are
    $A_{ij} = -1/\lambda_i$.  To compute its rank, multiply each row $i$
    with $\lambda_i$, and the new matrix has $n-1$ on the diagonal and 1
    everywhere else, hence the rank is $n-1$.  Finally, we check that
    there exists $\bm h^{(*)}$ such that
    $\bm d_i'' \cdot \bm h^{(*)} = 0$ for all $i$, by setting
    $\bm h^{(*)} = \sum_j \bm h^{(j)}$; then
    $\bm d_i'' \cdot \bm h^{(*)}$ is precisely the sum of the elements
    in row $i$ of the matrix $\bm A''$, hence it is $=0$.
\end{proof}

\section{Proof of Theorem~\ref{th:lambdas:real}}

\label{app:max:inequality}

We prove the theorem by generalizing it to arbitrary cones, in
Th.~\ref{th:lambdas:real:cone} below. Its proof, in turn, is based on
the following lemma.

\begin{lemma} \label{lemma:lambda:real} If $K \subseteq \R^m$ is a
  convex cone such that $K \cap (-\infty,0)^m = \emptyset$, then
  $K^* \cap \R_+^m \neq \set{0}$.
\end{lemma}
In other words, if a convex cone $K$ satisfies following property
  \begin{align}
      \forall \bm x \in K, &&\bigvee_{j\in [m]} x_j \geq & 0  \label{eq:property:m}
  \end{align}
  then there exists $\bm y$ s.t. $y_j \geq 0$ forall $j\in [m]$, and
  $\bm x \cdot \bm y \geq 0$ forall $\bm x \in K$.

\begin{proof}
  Let $\bm e_1, \ldots, \bm e_m$ be the canonical basis of $\R^m$,
  i.e.  $(\bm e_j)_i = \delta_{ij}$, and let
  $L = \cone(K \cup \set{\bm e_1, \ldots, \bm e_m})$.  We claim that
  $L$ also satisfies \eqref{eq:property:m}.  Indeed, every
  $\bm x' \in L$ has the form
  $\bm x' = \bm x + \sum_i \theta_i \bm e_i$ with $\bm x \in K$ and
  $\theta_i \geq 0$, $i=1,m$.  If $x'_j < 0$ for all $j$, then
  $x_j < 0$ for all $j$, which is a contradiction because $K$
  satisfies property \eqref{eq:property:m}.  Thus, $L$ is a convex
  cone, and disjoint from the strictly negative quadrant
  $(-\infty,0)^m$; since the latter is an open set, it is also
  disjoint from $\cl{L}$.  We claim that $L^* \neq \set{0}$.  Indeed,
  otherwise $L^{**} = \set{0}^* = \R^m$, but $L^{**}= \cl{L}$ is
  disjoint from $(-\infty,0)^m$, which is a contradiction.  Therefore
  there exists $\bm y \in L^*$ s.t. $\bm y \neq 0$.  Since
  $\bm e_i \in L$, we have $y_i = \bm y \cdot \bm e_i \geq 0$,
  i.e. $\bm y \in \R_+^m$.  Since $K \subseteq L$ we have
  $L^* \subseteq K^*$, hence $y \in K^*$ proving the lemma.
\end{proof}

We prove Theorem~\ref{th:lambdas:real} by generalizing it to arbitrary
convex cones.

\begin{theorem}\label{th:lambdas:real:cone} Let $S \subseteq \R^m$ be any
  convex cone, and $\bm y_1, \ldots, \bm y_k \in \R^m$.  Then the
  following two properties are equivalent:
  \begin{align}
    \forall \bm x \in S: & \max(\bm x \cdot \bm y_1, \ldots, \bm x \cdot  \bm y_k) \geq 0 \label{eq:lambda:1:cone}\\
    \exists \lambda_1\geq 0, \ldots, \lambda_k \geq 0, \sum_i \lambda_i=1,
    \forall \bm x \in S: & \sum_{i \in [k]} \lambda_i \bm x \cdot \bm y_i \geq 0 \label{eq:lambda:2:cone}
  \end{align}
\end{theorem}

Theorem~\ref{th:lambdas:real} is the special case when
$S \defeq \aentropic \cap \bigcap_{i \in [k]} \bm c_i^*$.

\begin{proof} (of Theorem~\ref{th:lambdas:real:cone}).  We prove only
  the implication from \eqref{eq:lambda:1:cone} to
  \eqref{eq:lambda:2:cone}; the other direction is immediate.  Define:
\begin{align*}
  K \defeq & \setof{(\bm x \cdot \bm y_1, \ldots, \bm x \cdot \bm y_k)}{\bm x \in S}
\end{align*}
Since $S$ is a convex cone, it follows that $K$ is also a convex cone,
hence we can apply Lemma~\ref{lemma:lambda:real} to $K$.  Condition
\eqref{eq:lambda:1} of the theorem states:
\begin{align*}
  \forall \bm z \in K: & \bigvee_{j \in [m]} \bm z_i \geq 0
\end{align*}
Lemma~\ref{lemma:lambda:real} implies that there exists
$\bm \lambda \in \R^k_+ \cap K^*$ s.t. $\bm \lambda \neq 0$:
\begin{align*}
  \forall \bm z \in K: & \bm \lambda \cdot \bm z \geq 0
\end{align*}
Expanding the definition of $K$, this condition becomes:
\begin{align}
  \forall \bm x  \in  S: & \sum_{i \in [k]} \lambda_i \bm x \cdot \bm y_i \geq 0 \label{eq:lambda:2:again}
\end{align}
We can assume w.l.o.g. that $\sum_i \lambda_i = 1$ (otherwise we
normalize it by dividing by $\sum_i \lambda_i$), and this proves
\eqref{eq:lambda:2:cone}.
\end{proof}

\section{Proof of Theorem~\ref{th:almost:recognizability:pi2}}

\label{app:almost:recognizability:pi2}

To prove the theorem, we need to establish the following
(separating-hyperplane-type) lemma:

\begin{lemma} \label{lemma:separating:hyperplane} Let
    $\bm h \in \R^{2^n}$ and suppose $\bm h \not\in \aentropic$.  Then
    there exists an information theoretic inequality with integral
    coefficients that is not satisfied by $\bm h$.  In other words,
    $\exists \bm c \in \Z^{2^n}$, such that
    $\forall \bm h_0 \in \entropic, \bm c \cdot \bm h_0 \geq 0$, and
    $\bm c \cdot \bm h < 0$.
\end{lemma}

To prove the lemma, we review some background, following
Studen{\'{y}}~\cite{DBLP:journals/kybernetika/Studeny93}.

\begin{definition}
    Given a cone $L$, define its {\em plane} as
    $pl(L) \defeq L \cap (-L)$.
\end{definition}


\begin{fact} \cite[Fact 9]{DBLP:journals/kybernetika/Studeny93} If $L$
    is a closed, convex cone, then $pl(L)$ is a vector space.
\end{fact}


\begin{definition}~\cite[Def. 6]{DBLP:journals/kybernetika/Studeny93} A
    closed convex cone $L \subseteq \R^m$ is called {\em regular} if
    $\Q$ is dense in $pl(L)$, i.e.  $\cl{\Q^m \cap pl(L)} = pl(L)$.
\end{definition}

\begin{lemma}~\cite[Prop. 3]{DBLP:journals/kybernetika/Studeny93} \label{lemma:regular}
    If $L$ is regular, then $Q^m$ is dense in $L^*$.
\end{lemma}

These lemmas allows us to prove the separating-plane lemma:

\begin{proof} (of Lemma~\ref{lemma:separating:hyperplane}) We use
    Lemma~\ref{lemma:regular}.  Every pointed cone $L$ is regular,
    because $pl(L)=\set{0}$.  The cone of almost entropic functions
    $K = \aentropic$ is pointed, because for every $\bm h \in K$ and
    every $X \subseteq [n]$, $h(X) \geq 0$; thus, if
    $\bm h_1, \bm h_2 \in \aentropic$ and $\bm h_1 + \bm h_2 = 0$, then
    $\bm h_1=\bm h_2=0$.  Therefore $K$ is regular, hence $\Q^{2^n}$ is
    dense in $K^*$.  Consider a vector that is not almost-entropic,
    $h \not\in K$.  Then there exists an information inequality
    $c_0 \in K^*$ s.t. $c_0 \cdot \bm h < 0$.  Since $\Q$ is dense in
    $K^*$, there exists a sequence in $\Q^{2^n} \cap K^*$ that converges
    to $c$, and therefore there exists some $c_1 \in \Q^{2^n} \cap K^*$
    s.t. $c_1 \cdot \bm h < 0$.  Multiply $c_1$ with the product of all
    $2^n$ denominators of its components, and obtain a vector
    $c \in \Z^{2^n} \cap K^*$ s.t. $c \cdot \bm h < 0$.  This completes
    the proof.
\end{proof}

And, finally, we use it to place the almost-entropic-recognizability
problem in $\Pi^0_2$:

\begin{proof} (of Theorem~\ref{th:almost:recognizability:pi2})
    Given $\bm c \in \Z^{2^n}$, let $P(\bm c)$ be the following
    predicate:
    \begin{align*}
    P(\bm c): && \forall \bm h\in \R^{2^n} (\bm h \in   \aentropic \Rightarrow \bm c \cdot \bm h \geq 0)
    \end{align*}
    Thus, $P(\bm c)$ checks if $\bm c$ defines a valid information
    inequality, and this, by Theorem~\ref{th:co:re}, is in $\Pi^0_1$.
    By Lemma~\ref{lemma:separating:hyperplane}, the almost-entropic
    recognizability problem is as follows.  Given $\bm h \in \R^{2^n}$
    (represented as in Def.~\ref{def:recognizability}):
    \begin{align*}
    \bm h \in \aentropic  \Leftrightarrow \forall \bm c (P(\bm c) \rightarrow \bm c \cdot \bm h \geq 0)
    \Leftrightarrow \forall \bm c (\neg P(\bm c) \vee \bm c \cdot \bm h \geq 0)
    \end{align*}
    which places the problem in $\Pi^0_2$ because $\neg P(\bm c)$ is in
    $\Sigma^0_1$.
\end{proof}

}

\end{document}